\documentclass[conference,letterpaper]{IEEEtran}

\usepackage[utf8]{inputenc} 
\usepackage[T1]{fontenc}
\usepackage{url}
\usepackage{ifthen}
\usepackage[cmex10]{amsmath} 

\usepackage{amsmath,amsthm,amssymb,lipsum}
\usepackage{enumerate}

\usepackage{graphicx}
\usepackage{fullpage,amsfonts,latexsym,amssymb,graphics}
\usepackage{tikz}
\usetikzlibrary{shapes,arrows,shadows}
\usepackage{amsmath,amstext,url}
\usepackage{pgfplots}
\pgfplotsset{width=10cm,compat=1.9}
\usetikzlibrary{intersections}
\newtheorem{theorem}{Theorem}

\newtheorem{result}{Result}
\newtheorem{example}{Example}
\newtheorem{corollary}{Corollary}

\usepackage{caption}
\usepackage{subcaption}
\usepackage{color}
\usepackage{textcomp}

\allowdisplaybreaks[3]
\tikzset{
  cross/.style   ={draw, cross out,inner sep=1mm, node distance = 2mm},
  triangle/.style   ={draw, regular polygon, regular polygon sides=3,inner sep=0.5mm, node distance = 2mm}
}

\newcommand{\bx}{\mathbf{x}}
\newcommand{\bz}{\mathbf{z}}
\newcommand{\by}{\mathbf{y}}
\newcommand{\bs}{\mathbf{s}}
\newcommand{\bn}{\mathbf{n}}

\newcommand{\bu}{\mathbf{u}}

\newcommand{\bt}{\mathbf{t}}

\newcommand{\cT}{\cal{T}}
\newcommand{\cN}{\cal{N}}

\newcommand{\cB}{\cal{B}}
\newcommand{\cM}{\cal{M}}

\makeatletter
\newcommand{\vast}{\bBigg@{4}}
\newcommand{\Vast}{\bBigg@{3.5}}
\makeatother

\begin{document}
\title{Channel Conditions for the Optimality of Interference Decoding Schemes for $K$-user Gaussian Interference Channels} 


\author{\IEEEauthorblockN{Ragini Chaluvadi, Bolli Madhuri and Srikrishna Bhashyam}
\IEEEauthorblockA{Department of Electrical Engineering\\
Indian Institute of Technology Madras\\
Chennai 600036, India\\
Email: skrishna@ee.iitm.ac.in }
}


\maketitle

\begin{abstract}
The sum capacity of the general $K$-user Gaussian Interference Channel (GIC) is known only when the channel coefficients are such that treating interference as noise (TIN) is optimal. The Han-Kobayashi (HK) scheme achieves the best known achievable rate region for the $K$-user interference channel (IC). Simple HK schemes are HK schemes with Gaussian signaling, no time sharing, and no private-common power splitting. The class of simple HK (S-HK) schemes includes the TIN scheme and schemes that involve various levels of interference decoding and cancellation at each receiver. We derive conditions under which simple HK schemes achieve sum capacity for general $K$-user Gaussian ICs. These results generalize existing sum capacity results for the TIN scheme to the class of simple HK schemes. 
\end{abstract}


\section{Introduction}
The capacity region and sum capacity of the general $K$-user Gaussian Interference Channel (GIC) are not known. The 2-user GIC is the most well understood special case \cite{car75,Sat81,annvee09,shakrache09,ShaKraChe08,motkha09,ettswa08}. 
The capacity region of the 2-user GIC under strong interference conditions was obtained in \cite{car75,Sat81}. The sum capacity when the interference can be treated as noise was obtained in \cite{annvee09,shakrache09,ShaKraChe08,motkha09}. The sum capacity under mixed interference conditions was obtained in \cite{motkha09}. The capacity region of the 2-user GIC within one bit was derived in \cite{ettswa08} using suitably chosen Han-Kobayashi (HK) schemes \cite{HanKob81}.

To the best of our knowledge, the sum capacity of the general $K$-user GIC is known only in the regime where Treating Interference as Noise (TIN) is optimal. For the general $K$-user GIC, the channel conditions under which TIN achieves sum capacity were obtained in \cite[Thm. 3]{ShaKraChe08}\cite[Thm. 9]{shang2008}. The sum capacity of some {\em partially connected} $K$ user GICs were derived in \cite{Tun11,liuerk11,PraBhaCho16,GnaChaBha17} under some channel conditions. $Z$-like GICs, where the channel matrix is upper triangular with a specific structure, were studied in \cite{Tun11}, cascade GIC was studied in \cite{liuerk11}, and many-to-one and one-to-many GICs were studied in \cite{PraBhaCho16,GnaChaBha17}. Some new outer bounds on the capacity of the $K$-user GIC were recently derived in \cite{Nam17}. {\em Simple} HK (S-HK)
schemes with Gaussian signalling, no timesharing, and no common-private power splitting, achieve sum capacity under the channel conditions obtained in \cite{Tun11,liuerk11,PraBhaCho16,GnaChaBha17}. S-HK schemes include the simple and practical TIN scheme and schemes that involve various levels of interference decoding and cancellation at each receiver as special cases. However, sum capacity results are available only for the TIN scheme for the fully-connected $K$-user GIC.

In this paper, we generalize the sum capacity optimality results for the TIN scheme in \cite{ShaKraChe08,shang2008}, to all S-HK schemes. In particular, we derive two sets of channel conditions under which S-HK schemes are sum capacity optimal for general $K$-user GICs. Exisiting results for the sum capacity of the 2-user GIC and some partially connected $K$ user GICs  in \cite{liuerk11,PraBhaCho16,GnaChaBha17} can be obtained as special cases of these results. Using Monte Carlo simulations, we evaluate the probability that these channel conditions for sum capacity are satisfied for some random wireless networks and observe that this probability is significant.

\section{Channel model and simple HK schemes}
\label{sec:model}
The $K$-user GIC in standard form \cite{ShaKraChe08} is given by
\begin{equation}
y_i= x_i+\underset{ \substack{j = 1\\ j \neq i}}{\overset{K}{\sum}} h_{ij} x_j+z_i, \ \forall i \in [K] \triangleq \{1,\hdots, K\},
\label{eqn:GICmodel}
\end{equation}
where $x_i$ is transmitted by transmitter $i$, $y_i$ is received by receiver $i$, $  h_{ij}$ is the real channel coefficient from transmitter $j$ to receiver $i$ and $z_i \sim \mathcal{CN}(0,1)$ is the additive white Gaussian noise at receiver $i$. Let $P_i$ denote the transmit power constraint at transmitter $i$. As in \cite{liuerk11}, we call HK schemes with Gaussian signaling, no timesharing, and no common-private power splitting as simple HK schemes. Each S-HK scheme is specified by the sets $\{ I(1), I(2), \hdots, I(K)\}$, $I(i) \subseteq [K] \backslash \{i\}, \ \forall i$. In each such S-HK scheme, at receiver $i$, interference from transmitters $j \in I(i)$ are treated as noise and interference from transmitters $j \in D(i) \triangleq \{[K] \backslash \{I(i),i \} \}$ are decoded. For the TIN scheme, $I(i) = [K]\backslash \{i\}$, $\forall i$.

\section{Sum Capacity Results}
In this section, we derive two sets of channel conditions for the general $K$-user GIC under which sum capacity is achieved by S-HK schemes. The first set of channel conditions are in equations (\ref{eqn:fin1})-(\ref{eqn:fin3}) of Theorem \ref{thm:covgen}. The second set of channel conditions are given by equations (\ref{eqn:achsum-rategenIC}) and  (\ref{eqn:jointfin1})-(\ref{eqn:jointfin4}) in Theorems \ref{thm:achGIC} and \ref{thm:convjoint}, respectively.

In the result in Theorem \ref{thm:covgen}, we consider the strategy of decoding interference from transmitters in $D(i)$ for each $i$ before decoding the desired message. For such decoding to be possible, conditions in (\ref{eqn:fin3}) need to be satisfied. For the optimality of treating the interference from transmitters in $I(i)$ as noise for each $i$, we get conditions (\ref{eqn:fin1})-(\ref{eqn:fin2}). These conditions correspond to the TIN optimality conditions for the modified GIC where all the links corresponding to decoded interference are removed.
\begin{theorem} \label{thm:covgen}
For the $K$-user GIC, the S-HK scheme defined by $I(i) \subseteq [K] \backslash \{i\}, \ \forall i \in [K]$ achieves sum capacity, if there exist $\rho_i \in (0,1)$, $\forall i \in [K]$, such that the following conditions are satisfied for all $i \in [K]$
\begin{align}
&\sum_{j: i \in I(j)} \left[ \frac{h_{ji}^2}{1+Q_j-\rho_j^2}\right]  \leq \frac{1}{P_i+\left(\frac{1+Q_i}{\rho_i}\right)^2}, \label{eqn:fin1} \\
& \underset{j \in I(i)}{\sum} \frac{h_{ij}^2(1+Q_j)^2}{\rho_j^2}  \leq  1-\rho_i^2, \label{eqn:fin2}\\
& \underset{j \in \mathcal{J}}{\prod} \left( 1+\frac{P_j}{1+Q_j} \right)  \leq \left(1+ \frac{\sum_{j \in \mathcal{J}}h_{ij}^2P_j}{1+P_i+Q_i} \right) \forall \mathcal{J} \subseteq D(i), \label{eqn:fin3}
\end{align}
where $Q_i=\underset{j \in I(i)}{\sum} h_{ij}^2P_j $,  $ D(i)=[K] \backslash \{i, I(i)\}.$ The sum capacity is
\begin{equation}
C_{sum} = \underset{i=1}{\overset{K}{\sum}} \frac{1}{2} \log \left[1+\frac{P_i}{1+Q_i} \right]. \label{eqn:sumcapacitygen}
\end{equation}
\end{theorem}
\begin{proof}
(Achievability) Suppose that each receiver $i$ decodes the interference from transmitters $D(i)$, and then decodes the information from $i^{th}$ transmitter, while treating interference from other transmitters $I(i)$ as noise. The multiple access channel (MAC) constraints for decoding the interference at each receiver $i$ are 
\begin{align}
\sum_{j \in \mathcal{J}} R_j & \leq \frac{1}{2} \log \left(1+ \frac{\sum_{j \in \mathcal{J}}h_{ij}^2P_j}{1+P_i+Q_i}\right), \forall \mathcal{J} \subseteq D(i). \label{eqn:condition2}
\end{align}
The sum capacity in (\ref{eqn:sumcapacitygen}) is achieved if choosing
\begin{equation}
R_i = \frac{1}{2} \log\left(1+\frac{P_i}{1+Q_i} \right), \ \forall i \in [K] \label{eqn:condition1}
\end{equation}
satisfies (\ref{eqn:condition2}), thereby resulting in conditions in (\ref{eqn:fin3}).\\

(Converse) For each receiver $i \in [K]$, use the genie signal $\bs_i^n= \{\bx_i^n+\bn_i^n, \  \bx_j^n, \ j\in D(i)\}$ where $\bn_i^n \sim \mathcal{N}(\mathbf{0}, \sigma_i^2 \mathbf{I})$ and $E[n_i z_i]=\rho_i \sigma_i$, $0 < \rho_i < 1$. Here, for each $i$, we provide signals $\bx_j^n, \forall j \in D(i)$ in addition to the genie signal $\bx_i^n+\bn_i^n$ that is used in \cite{shang2008}. Under (\ref{eqn:fin1}) and (\ref{eqn:fin2}), we get the required upper bound following steps similar to the proof in \cite[Theorem 9]{shang2008}, but with the above genie signals. The details are provided in Appendix \ref{app:convthm1}.

Combining the conditions (\ref{eqn:fin1}) and (\ref{eqn:fin2}) for the converse with the conditions (\ref{eqn:fin3}) for achievability, we get the required result.
\end{proof}

Now, we derive the second set of channel conditions under which S-HK schemes are optimal. We do this in two steps. First, we derive general bounds on the achievable sum rate of S-HK schemes in Theorem \ref{thm:achGIC}. Unlike Theorem \ref{thm:covgen}, where the interference is decoded and cancelled before decoding the desired signal, here we determine more general bounds on the achievable sum rate for an S-HK scheme. Then, we show in Theorem \ref{thm:convjoint} that one of the sum rate upperbounds in Theorem \ref{thm:achGIC} is also a sum capacity bound under some channel conditions. Therefore, the channel conditions under which we get a sum capacity result will comprise of (i) the conditions (\ref{eqn:jointfin1})-(\ref{eqn:jointfin4}) required to prove the sum capacity upper bound in Theorem \ref{thm:convjoint}, and (ii) the conditions (\ref{eqn:achsum-rategenIC}) under which this sum capacity upperbound is achievable in Theorem \ref{thm:achGIC}. 
\begin{theorem}\label{thm:achGIC}
For the $K$-user GIC, the S-HK scheme defined by $\{I(i)\}$ achieves sum rates $S$ satisfying the following conditions for each $l \in [K]$. 
\begin{IEEEeqnarray}{lcr}
l.S  \leq \frac{1}{2} \sum_{i \in [K]} \log \left(1+\frac{\sum_{j \in \mathcal{J}_i} h_{ij}^2P_j}{1+ Q_i} \right)
\label{eqn:achsum-rategenIC}
\end{IEEEeqnarray}
for each choice of ${\cal J}_i \subseteq [K]\backslash I(i)$ such that  $\underset{i \in {[K]}}{\bigcup} \mathcal{J}_i= {\cal S}_l$. Here ${\cal S}_l$ is a multiset containing $l$ copies of each element in $[K]$ and is denoted ${\cal S}_l = \{ (a,l): a \in [K] \}$, and $Q_i=\underset{j \in I(i)}{\sum} h_{ij}^2P_j$.
\end{theorem}
\begin{proof}
At each receiver $i$, users $[K] \backslash I(i)$ form a Gaussian MAC with noise variance $1+Q_i$. The achievable rates of each MAC at receiver $i \in [K]$ satisfy 
\begin{align*}
\sum_{j \in \mathcal{J}_i} R_j & \leq \frac{1}{2} \log \left(1+\frac{\sum_{j \in \mathcal{J}_i} h_{ij}^2P_j}{1+ Q_i} \right) \forall \mathcal{J}_i \subseteq [K] \backslash I(i). 
\end{align*}
Using Fourier-Motzkin elimination, we get the sum rate bounds in (\ref{eqn:achsum-rategenIC}). 
\end{proof}
The maximum sum rate achievable using an S-HK scheme is determined by the least upper bound for $S$ among the bounds in (\ref{eqn:achsum-rategenIC}).
As an example of a bound in the above theorem, consider $l = 1$, ${\cal J}_m = \{m, k\}$ for some $m, k \in [K]$, ${\cal J}_k = \phi$, and ${\cal J}_i = i$ for $i \in [K] \backslash \{m, k\}$. This gives us the bound on sum rate to be
\begin{IEEEeqnarray*}{lcr}
\frac{1}{2} \log \left[1+\frac{P_m+h_{mk}^2P_k}{1+Q_m} \right] + \underset{i \neq k,m}{\underset{i=1}{\overset{K}{\sum}}} \frac{1}{2} \log \left[1+\frac{P_i}{1+Q_i} \right]. 
\end{IEEEeqnarray*}
Now, if we can show that one of these inequalities in (\ref{eqn:achsum-rategenIC}) is also an upper bound on the sum capacity under some conditions, then we get a sum capacity result.  In the following theorem, we show that the sum rate bound expression in the example above is a sum capacity upper bound under conditions (\ref{eqn:jointfin1})-(\ref{eqn:jointfin4}) (for the choice $G(i) = I(i)$ in the following theorem).
\begin{theorem} \label{thm:convjoint}
Let $G(i) \subseteq [K]\backslash \{i\}, \ \forall i \in [K]$ and let there be some $m, k \in [K]$ such that $m, k \notin G(i)$, $\forall i \in [K] \backslash \{k\}$. For the $K$-user GIC, if there exist $\rho_i \in (0,1)$, $\forall i \in  [K]\backslash \{m\}$ such that the following conditions are satisfied 
\begin{align}
&\frac{1}{P_r+\left(\frac{1+Q_r}{\rho_r}\right)^2}  \geq  \sum_{\underset{i \neq \{m,k\}}{i: r \in G(i)}} \left[ \frac{h_{ir}^2}{1+Q_i-\rho_i^2}\right]  \nonumber\\
& ~ + \delta_r \left[\frac{h_{mr}^2}{1+Q_m-\rho_k^2} \right] \forall \ r \in [K]\backslash \{m,k\}, \label{eqn:jointfin1} \\
&\rho_k h_{mk} = 1+ Q_m \label{eqn:jointfin2}\\
 &\underset{j \in G(i)}{\sum} \frac{h_{ij}^2(1+Q_j)^2}{\rho_j^2} \leq   1-\rho_i^2,  \forall i \in [K] \backslash \{m,k\}, \label{eqn:jointfin3}\\
   & \underset{j \in G(m)}{\sum} \frac{h_{mj}^2(1+Q_j)^2}{\rho_j^2} \leq 1-\rho_k^2, \label{eqn:jointfin4}
\end{align}
where  $\delta_r = 1$ if $r \in G(m)$ and $\delta_r = 0$ otherwise, and $Q_i=\underset{j \in G(i)}{\sum} h_{ij}^2P_j $, then the sum capacity $C_{sum}$ is upper bounded by 
\begin{IEEEeqnarray*}{lcr}
& \frac{1}{2} \log \left[1+\frac{(P_m+h_{mk}^2P_k)}{1+Q_m} \right]  + \underset{i \neq k,m}{\underset{i=1}{\overset{K}{\sum}}} \frac{1}{2} \log \left[1+\frac{P_i}{1+Q_i} \right]. 
\end{IEEEeqnarray*}
\end{theorem}
\begin{proof}
The detailed proof is provided in Appendix \ref{app:proofthm3}. Here, we present a brief outline and highlight some aspects of the proof.
First, we consider a modified channel with no interference at receiver $k$. The sum capacity of the original channel is upper bounded by the sum capacity of the modified channel. 
Then, we derive a genie-aided upper bound for the modified channel using the 
genie signals $\bs_i^n$ at receiver $i$ for each $i \in [K]$ as follows:
\begin{align*}
\bs_i^n & = \{\bx_i^n+\bn_i^n, \bx_j^n, j \in \bar{G}(i)\} , \ \forall i\in [K] \backslash \{m,k\} \\
\bs_m^n & =\{\bx_j^n,\  j \in \bar{G}(m) \backslash k \}\\
\bs_k^n &=h_{mk}\bx_k^n+\underset{j \in G(m)}{\sum} h_{mj} \bx_j^n+ \bn_k^n
\end{align*}
where $\bar{G}(i)=[K] \backslash \{i,\{G(i)\}\}$, $n_i \sim \mathcal{N}(0, \sigma_i^2)$, $E[n_i z_i]=\rho_i \sigma_i$, and  $0 < \rho_i < 1$, for each $i \in [K] \backslash \{m\}$, and $\sigma_k=1$.
This choice of genie is then shown to be useful and smart under conditions (\ref{eqn:jointfin1})-(\ref{eqn:jointfin4}) to obtain the upper bound in the theorem statement. 

Here are some remarks about this proof.
\begin{itemize}
\item The genie signal is different from Theorem 1 for receivers $m$ and $k$. The genie at receiver $k$ has the interference component at receiver $m$ from transmitter $k$ and the other transmitters that are treated as noise. This choice ensures that $h(\bs_k^n)=h(\by_m^n |\bs_m^n, \bx_m^n )$ and helps in cancelling one negative term in the sum capacity upper bound.
\item The assumption that  $m, k \notin G(i)$, $\forall i \in [K] \backslash \{k\}$ is used as part of the argument that the genie is useful.
\item The first upper bounding step is with a modified channel with no interference at receiver $k$. It is interesting to note that the sum capacity result for the 2-user GIC under mixed interference in \cite{motkha09} also uses the one-sided GIC as the first step, and we recover these results as special cases of our result.
\item This proof also generalizes the proof for the many-to-one GIC in \cite[Theorem 4]{GnaChaBha17} to the general $K$ user GIC.
\end{itemize}

\end{proof}

Some examples of the conditions obtained from Theorems 1-3 are presented in Appendix \ref{app:examples}.

\subsection{Relation with exisiting sum capacity results}
Applying Theorems \ref{thm:covgen}-\ref{thm:convjoint} to the special case of 2-user channels, i.e., $K = 2$, we recover all known sum capacity results for the 2-user GIC in \cite{car75,Sat81,annvee09,shakrache09,motkha09}. The first set of channel conditions in our paper gives the noisy interference result in \cite{annvee09,shakrache09,motkha09}, the very strong interference result in \cite{car75}, and part of the mixed interference result in \cite[Thm. 10]{motkha09} as special cases. The second set of channel conditions in our paper gives the remaining part of the mixed interference result in \cite[Thm. 10]{motkha09} and the strong interference result in \cite{Sat81}. The actual list of channel conditions and the corresponding sum capacity are in Appendix \ref{app:2user}.

Applying Theorems \ref{thm:covgen}-\ref{thm:convjoint} to the special cases of partially connected Gaussian ICs, we can recover the sum capacity results in \cite{GnaChaBha17,PraBhaCho16,liuerk11}. We can also get some new results for the $K$-user cyclic and cascade GICs. The results corresponding to the two channel conditions for the cyclic, cascade and many-to-one GICs are presented in Appendix \ref{app:partialres}.

\section{Numerical Results}
In this section, we numerically find the probability that the first set of channel conditions under which S-HK schemes achieve sum capacity, i.e, equations (\ref{eqn:fin1})-(\ref{eqn:fin3}), are satisfied for three different random wireless network topologies. A similar study for the second set of channel conditions is currently under progress.

\textit{Topology 1}: In this topology, all $K$ transmitters are placed randomly and uniformly in a circular cell of radius 1 km. We assume that each transmitter has a nominal coverage radius of $r_1$ m. For each transmitter, we then place its receiver randomly and uniformly in its coverage area. This topology is illustrated in figure \ref{uniform_setup} for $K=5$.
\begin{figure}
\centering
\begin{tikzpicture}[auto,thick,scale=0.3]
\draw[color=red!80] (0,0) circle [radius=10cm];
\draw[color=blue,dashed] (5,-2) circle [radius=3.5cm];
\draw[color=blue,dashed] (-5,4) circle [radius=3.5cm];
\draw[color=blue,dashed] (-4,-6) circle [radius=3.5cm];
\draw[color=blue,dashed] (7,7) circle [radius=3.5cm];
\draw[color=blue,dashed] (2,2) circle [radius=3.5cm];

\draw 
      node at (5,-2) [triangle, blue, name=t1] {}
      node at (-5,4) [triangle,blue, name=t2] {}
      node at (-4,-6) [triangle,blue, name=t3] {}
      node at (7,7) [triangle,blue, name=t4] {}
      node at (2,2) [triangle,blue, name=t5] {};
\draw 
      node at (7.6,-2.7) [cross, red, name=r1] {}
      node at (-6,2) [cross, red, name=r2] {}
      node at (-3,-5) [cross, red, name=r3] {}
      node at (5,5) [cross, red, name=r4] {}
      node at (3.5,2.5) [cross, red, name=r5] {};
      
\draw 
      node[thick] at (4,-2.5)  {\small $T_1$}
      node[thick] at (-4.5,5) {\small $T_2$}
      node[thick] at (-3,-6.5) {\small $T_3$}
      node[thick] at (7.9,7.5) {\small $T_4$}
      node[thick] at (1,1.5) {\small $T_5$};
\draw 
      node[thick] at (7.5,-1.9) {\small $R_1$}
      node[thick] at (-6,3) {\small $R_2$}
      node[thick] at (-2.5,-4) {\small $R_3$}
      node[thick] at (4.9,5.9) {\small $R_4$}
      node[thick] at (4.5,2.5){\small $R_5$};

\draw [<->] (0,0) -- (-10,0);
\draw [<->] (-4,-6) -- (-7.5,-6);

\draw 
      node[thick] at (-5,-0.5)  { $1$ Km}
      node[thick] at (-5.5,-6.5) { $r_1$ };
\end{tikzpicture}
  \caption{Topology 1 setup where triangles are transmitters and crosses are receivers}
    \label{uniform_setup}
\end{figure}
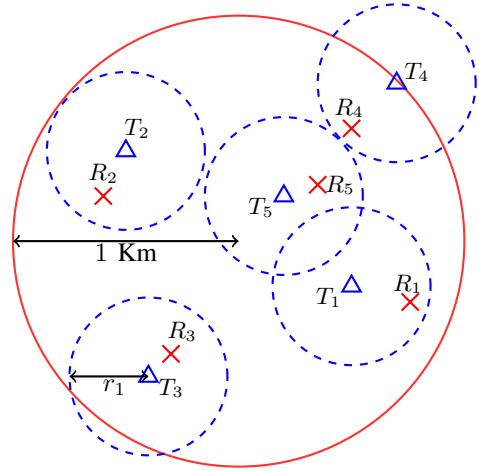

\textit{Topology 2}: In this topology, the first transmitter is placed at the center of a circle of radius $r_2$ m and all the other transmitters are placed equally spaced on the perimeter of this circle. The nominal coverage radius of first transmitter is $3 r_2$ m and nominal coverage radius of all other transmitters are $r_2$ m. For each transmitter, we place its receiver randomly and uniformly in its coverage area. This topology for $K=4$ is illustrated in figure \ref{onetomany_setup}. In topology 2, the first transmitter has a longer range and, therefore, there is higher probability that its signal at other receivers is strong enough to decode.\\
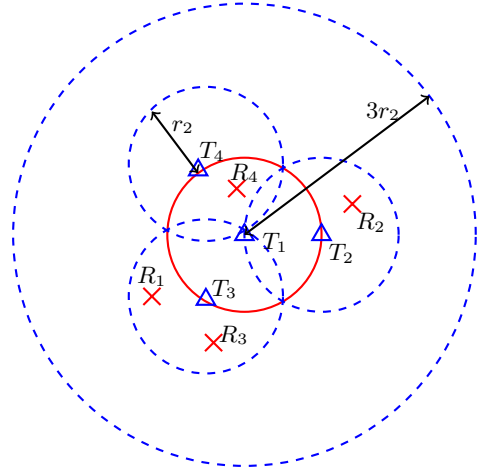
\begin{figure}
    \centering
\begin{tikzpicture}[auto,thick,scale=0.205]
\draw[color=blue,dashed] (-0,0) circle [radius=15cm];
\draw[color=red] (0,0) circle [radius=5cm];
\draw[color=blue,dashed] (5,0) circle [radius=5cm];
\draw[color=blue,dashed] (-2.5,4.6) circle [radius=5cm];
\draw[color=blue,dashed] (-2.5,-4) circle [radius=5cm];

\draw 
      node at (0,0) [triangle, blue, name=t1] {}
      node at (5,0) [triangle,blue, name=t2] {}
      node at (-2.5,-4.2) [triangle,blue, name=t3] {}
      node at (-3,4.2) [triangle,blue, name=t4] {};
\draw 
      node at (-6,-4) [cross, red, name=r1] {}
      node at (7,2) [cross, red, name=r2] {}
      node at (-2,-7) [cross, red, name=r3] {}
      node at (-0.5,3) [cross, red, name=r4] {};
\draw 
      node[thick] at (2,-0.5)  {\small $T_1$}
      node[thick] at (6.3,-1.2) {\small $T_2$}
      node[thick] at (-1.3,-3.5) {\small $T_3$}
      node[thick] at (-2.1,5.4) {\small $T_4$};
\draw 
      node[thick] at (-6,-2.7) {\small $R_1$}
      node[thick] at (8.2,0.8) {\small $R_2$}
      node[thick] at (-0.7,-6.5) {\small $R_3$}
      node[thick] at (-0,4) {\small $R_4$};

\draw [<->] (0,0) -- (12,9);
\draw [<->] (-3,4) -- (-6,8);

\draw 
      node[thick] at (9,8)  {\small $3r_2$}
      node[thick] at (-4,7) {\small $r_2$ };
\end{tikzpicture}
    \caption{Topology 2 setup where triangles are transmitters and crosses are receivers}
    \label{onetomany_setup}
\end{figure}
\textit{Topology 3}: In this topology, all transmitters are placed equidistantly along a line, with transmitter to transmitter distance $r_3$ m. For each transmitter, we place its corresponding receiver randomly and uniformly along the same line towards its right within $r_3$ m. We assume that the nominal coverage radius of each transmitter is $r_3$ m. This topology for $K=4$ is illustrated in figure \ref{cascade_setup}. In topology 3,  each receiver usually observes strong interference only from its adjacent transmitter.\\
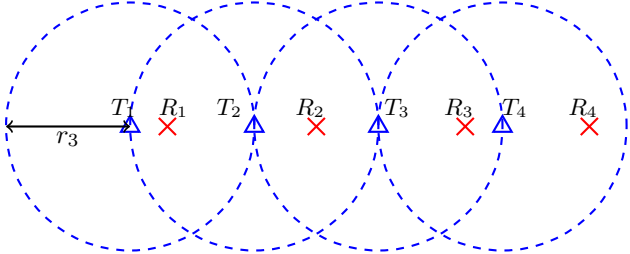
\begin{figure}
\begin{tikzpicture}[auto,thick,scale=0.165]

\draw[color=blue,dashed] (5,0) circle [radius=10cm];
\draw[color=blue,dashed] (15,0) circle [radius=10cm];
\draw[color=blue,dashed] (25,0) circle [radius=10cm];
\draw[color=blue,dashed] (35,0) circle [radius=10cm];

\draw 
      node at (5,0) [triangle, blue, name=t1] {}
      node at (15,0) [triangle,blue, name=t2] {}
      node at (25,0) [triangle,blue, name=t3] {}
      node at (35,0) [triangle,blue, name=t4] {};
\draw 
      node at (8,0) [cross, red, name=r1] {}
      node at (20,0) [cross, red, name=r2] {}
      node at (32,0) [cross, red, name=r3] {}
      node at (42,0) [cross, red, name=r4] {};
\draw 
      node[thick] at (4.5,1.5)  {\small $T_1$}
      node[thick] at (13,1.5) {\small $T_2$}
      node[thick] at (26.5,1.5) {\small $T_3$}
      node[thick] at (36,1.5) {\small $T_4$};
\draw 
      node[thick] at (8.5,1.5) {\small $R_1$}
      node[thick] at (19.5,1.5) {\small $R_2$}
      node[thick] at (31.5,1.5) {\small $R_3$}
      node[thick] at (41.5,1.5) {\small $R_4$};

\draw [<->] (-5,0) -- (5,0);

\draw 
      node[thick] at (0,-1) { $r_3$ };
\end{tikzpicture}
    \caption{Topology 3 setup where triangles are transmitters and crosses are receivers}
    \label{cascade_setup}
\end{figure}
For channel fading, we use the Erceg model \cite{ercgre99} as done in \cite{gennadavejaf15}. We used an operating frequency of 1.9 GHz and the terrain category of hilly/light tree density for model parameters. The noise floor is taken as $-110$ dBm and transmit power at each transmitter is chosen such that the expected value of the SNR at the boundary of their nominal coverage area is 0 dB.

For generating the plots, we consider 1000 realizations of the channel. With topology 1, for every realization we randomly place $K$ transmitters inside 1 km circular cell and also randomly place each receiver in its corresponding transmitters coverage area. With topology 2 and topology 3, first we fix the transmitters locations and for every realization we randomly place each receiver in its corresponding transmitter's coverage area.\\

In Figures \ref{fig:topology1TINDI}, \ref{fig:topology2TINDI}, \ref{fig:topology3TINDI}, we plot the probability that the conditions (\ref{eqn:fin1})-(\ref{eqn:fin3}) are satisfied for (i) TIN scheme, (ii) all S-HK schemes except the TIN scheme (denoted S-HK$\backslash$TIN) and (iii) all S-HK schemes. Figures \ref{fig:topology1TINDI}, \ref{fig:topology2TINDI}, \ref{fig:topology3TINDI} are plotted for topologies 1, 2, and 3, respectively. We observe that the probability that the conditions for optimality are satisfied is significant. In Fig. \ref{fig:topology1TINDI}, this probability increases with increasing nominal coverage radius $r_1$ as expected for $\text{S-HK}\backslash\text{TIN}$. 
\begin{figure}[ht]
\centering
\begin{tikzpicture}[scale=0.78]
\begin{axis}
[    grid=both,
    grid style={line width=.1pt, draw=gray!10},
    major grid style={line width=.2pt,draw=gray!20},xlabel=Coverage radius $r_1$ of each transmitter, ylabel=Probability,
cycle list={%
black,mark=o\\%
blue,mark=triangle\\%
red, mark=square\\%
}
]

\addplot coordinates {
(10,0.996)(
110,0.923)(
210,0.749)(
310,0.547)(
410	,0.372)(
510,0.265)(
610	,0.136)(
710	,0.093)(
810	,0.073)(
910	,0.038)(
1000 ,0.022)

  } ;
\addplot coordinates {
(10	,0.002)(
110	,0.025)(
210	,0.067)(
310	,0.13)(
410	,0.131)(
510	,0.148)(
610	,0.159)(
710	,0.14)(
810	,0.121)(
910	,0.115)(
1000,0.114)

  } ;
  
\addplot coordinates {
( 10,0.998)(
110	,0.948)(
210	,0.816)(
310,0.677)(
410,0.503)(
510	,0.413)(
610,0.295)(
710	,0.233)(
810	,0.194)(
910	,0.153)(
1000,0.136)
};
\addlegendentry{TIN}
\addlegendentry{S-HK$\backslash$TIN}
\addlegendentry{S-HK}
\end{axis}
\end{tikzpicture}
\caption{Success probability of conditions (\ref{eqn:fin1}), (\ref{eqn:fin2}), (\ref{eqn:fin3}) for TIN scheme , S-HK schemes excluding TIN and all S-HK schemes with topology 1, $K=3$ }
\label{fig:topology1TINDI}
\end{figure}

\begin{figure}[ht]
\centering
\begin{tikzpicture}[scale=0.78]
\begin{axis}
[    grid=both,
    grid style={line width=.1pt, draw=gray!10},
    major grid style={line width=.2pt,draw=gray!20},xlabel=Coverage radius $r_2$ of transmitters 2 and 3, ylabel=Probability,
legend style={at={(0.6,0.5)},anchor=west},
cycle list={%
black,mark=o\\%
blue,mark=triangle\\%
red, mark=square\\%
}
]

\addplot coordinates {
(
160	,0.004)(
210	,0.006)(
260	,0.011)(
310	,0.007)(
360	,0.007)(
410	,0.005)(
460	,0.002)

  } ;
\addplot coordinates {
(
160	,0.316)(
210	,0.31)(
260	,0.298)(
310	,0.278)(
360	,0.295)(
410	,0.288)(
460	,0.298)

  } ;
  
\addplot coordinates {
(
160	,0.32)(
210	,0.316)(
260	,0.309)(
310	,0.285)(
360	,0.302)(
410	,0.293)(
460	,0.3)
};
\addlegendentry{TIN}
\addlegendentry{S-HK$\backslash$TIN}
\addlegendentry{S-HK}
\end{axis}
\end{tikzpicture}
\caption{Success probability of conditions (\ref{eqn:fin1}), (\ref{eqn:fin2}), (\ref{eqn:fin3}) for TIN scheme , S-HK schemes excluding TIN and all S-HK schemes with topology 2, $K=3$ }
\label{fig:topology2TINDI}
\end{figure}
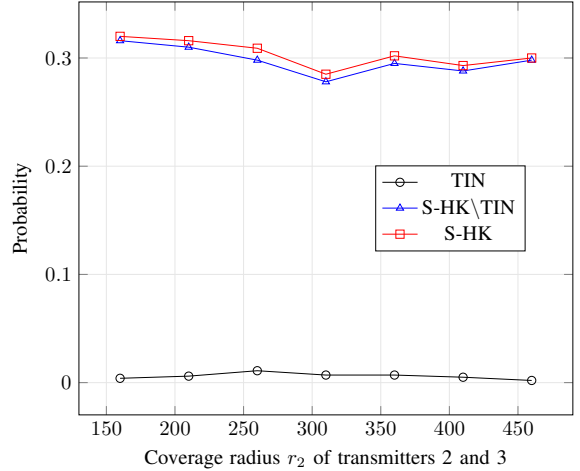

\begin{figure}[ht]
\centering
\begin{tikzpicture}[scale=0.78]
\begin{axis}
[ tick label style={/pgf/number format/fixed},
 grid=both,
    grid style={line width=.1pt, draw=gray!10},
    major grid style={line width=.2pt,draw=gray!20},xlabel=Coverage radius $r_3$ of each transmitter, ylabel=Probability,
legend style={at={(0.35,0.8)},anchor=east},
cycle list={%
black,mark=o\\%
blue,mark=triangle\\%
red, mark=square\\%
}
]

\addplot coordinates {
(10	,0.015)(
60	,0.009)(
110	,0.021)(
160	,0.044)(
210	,0.041)(
260	,0.076)(
310	,0.07)(
360	,0.071)(
410	,0.058)(
460	,0.078)

  } ;
\addplot coordinates {
(10,0)(
60	,0.018)(
110	,0.102)(
160	,0.101)(
210	,0.123)(
260	,0.117)(
310	,0.144)(
360	,0.138)(
410	,0.155)(
460	,0.149)

  } ;
  
\addplot coordinates {
(10	,0.015)(
60	,0.027)(
110	,0.123)(
160	,0.145)(
210	,0.164)(
260	,0.193)(
310	,0.214)(
360	,0.209)(
410	,0.213)(
460	,0.227)

};
\addlegendentry{TIN}
\addlegendentry{S-HK$\backslash$TIN}
\addlegendentry{S-HK}
\end{axis}
\end{tikzpicture}
\caption{Success probability of conditions (\ref{eqn:fin1}), (\ref{eqn:fin2}), (\ref{eqn:fin3}) for TIN scheme , S-HK schemes excluding TIN and all S-HK schemes with topology 3, $K=3$  }
\label{fig:topology3TINDI}
\end{figure}

In Figures \ref{fig:topology1TINDIspl}, \ref{fig:topology2TINDIspl}, \ref{fig:topology3TINDIspl}, we plot the probability that the conditions (\ref{eqn:fin1}), (\ref{eqn:fin2}), (\ref{eqn:fin3}) are satisfied for (i) all S-HK schemes except TIN scheme, (ii) all S-HK schemes where atmost 1 strong interference signal is decoded at each receiver except TIN (denoted S-HK1$\backslash$TIN). Figures \ref{fig:topology1TINDIspl}, \ref{fig:topology2TINDIspl}, \ref{fig:topology3TINDIspl} are plotted for topologies 1, 2, and 3, respectively. In topologies 2 and 3, decoding atmost one strong interference at each receiver is the most important class of S-HK schemes as expected since there is mainly one strongly interfering signal in these topologies.

\begin{figure}[ht]
\centering
\begin{tikzpicture}[scale=0.78]
\begin{axis}
[ tick label style={/pgf/number format/fixed},ymin=0, ymax=0.2,  
  grid=both,
    grid style={line width=.1pt, draw=gray!10},
    major grid style={line width=.2pt,draw=gray!20},xlabel=Coverage radius $r_1$ of each transmitter, ylabel=Probability,
cycle list={%
blue,mark=*\\%
black,mark=triangle\\%
red, mark=square\\%
}
]

\addplot coordinates {
(410,0.131)(
510	,0.148)(
610	,0.159)(
710	,0.14)(
810	,0.121)(
910	,0.115)(
1000,0.114)

  } ;
\addplot coordinates {
(410,0.124)(
510	,0.136)(
610	,0.13)(
710,0.118)(
810	,0.09)(
910	,0.079)(
1000,0.062)

  } ;
  
\addlegendentry{S-HK$\backslash$TIN}
\addlegendentry{S-HK1$\backslash$TIN}
\end{axis}
\end{tikzpicture}
\caption{Success probability of conditions (\ref{eqn:fin1}), (\ref{eqn:fin2}), (\ref{eqn:fin3}) for topology 1, $K=3$ }
\label{fig:topology1TINDIspl}
\end{figure}

\begin{figure}[ht]
\centering
\begin{tikzpicture}[scale=0.78]
\begin{axis}
[  ymin=0, ymax=0.4, 
grid=both,
    grid style={line width=.1pt, draw=gray!10},
    major grid style={line width=.2pt,draw=gray!20},xlabel=Coverage radius $r_2$ of transmitters 2 and 3, ylabel=Probability,
legend style={at={(0.5,0.4)},anchor=west},
cycle list={%
blue,mark=*\\%
black,mark=triangle\\%
red, mark=square\\%
}
]

\addplot coordinates {
(160,0.316)(
210	,0.31)(
260	,0.298)(
310	,0.278)(
360	,0.295)(
410,0.288)

  } ;
\addplot coordinates {
(160,0.296)(
210	,0.29)(
260	,0.256)(
310	,0.243)(
360	,0.267)(
410	,0.244)

  } ;
  
\addlegendentry{S-HK$\backslash$TIN}
\addlegendentry{S-HK1$\backslash$TIN}
\end{axis}
\end{tikzpicture}
\caption{Success probability of conditions (\ref{eqn:fin1}), (\ref{eqn:fin2}), (\ref{eqn:fin3}) for topology 2, $K=3$ }
\label{fig:topology2TINDIspl}
\end{figure}

\begin{figure}[ht]
\centering
\begin{tikzpicture}[scale=0.78]
\begin{axis}
[ tick label style={/pgf/number format/fixed},ymin=0, ymax=0.2,  grid=both,
    grid style={line width=.1pt, draw=gray!10},
    major grid style={line width=.2pt,draw=gray!20},xlabel=Coverage area $r_3$ of each transmitter, ylabel=Probability,
legend style={at={(0.1,0.8)},anchor=west},
cycle list={%
blue,mark=*\\%
black,mark=triangle\\%
red, mark=square\\%
}
]

\addplot coordinates {
(160	,0.101)(
210	,0.123)(
260	,0.117)(
310	,0.144)(
360	,0.138)(
410	,0.155
)

  } ;
\addplot coordinates {
(160	,0.084)(
210	,0.1)(
260	,0.1)(
310	,0.131)(
360	,0.127)(
410	,0.138
)
  } ;
  
\addlegendentry{S-HK$\backslash$TIN}
\addlegendentry{S-HK1$\backslash$TIN}
\end{axis}
\end{tikzpicture}
\caption{Success probability of conditions (\ref{eqn:fin1}), (\ref{eqn:fin2}), (\ref{eqn:fin3}) for topology 3, $K=3$ }
\label{fig:topology3TINDIspl}
\end{figure}

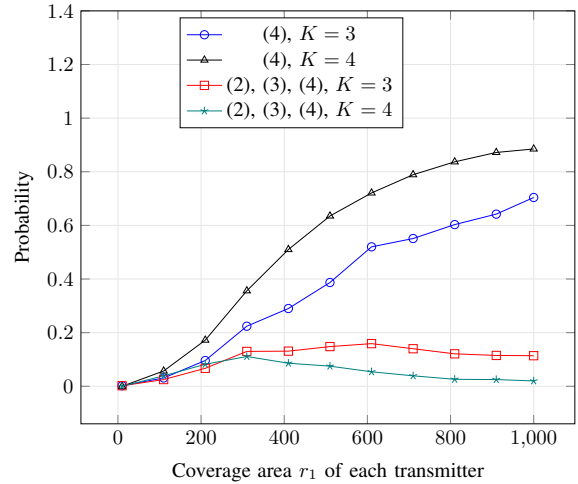
\begin{figure}[ht]
\centering
\begin{tikzpicture}[scale=0.78]
\begin{axis}
[ymax=1.4,    
grid=both,
    grid style={line width=.1pt, draw=gray!10},
    major grid style={line width=.2pt,draw=gray!20},xlabel=Coverage area $r_1$ of each transmitter, ylabel= Probability,
legend style={at={(0.2,0.85)},anchor=west},
cycle list={%
blue,mark=o\\%
black,mark=triangle\\%
red, mark=square\\%
teal,mark=star\\%
}
]


\addplot coordinates {
(10	,0.002)(
110	,0.031)(
210	,0.096)(
310	,0.224)(
410	,0.29)(
510,	0.387)(
610	,0.52)(
710	,0.551)(
810	,0.603)(
910	,0.642)(
1000,0.704)
  } ;

\addplot coordinates {
(10	,0)(
110	,0.057)(
210	,0.172)(
310	,0.356)(
410	,0.51)(
510	,0.635)(
610	,0.721)(
710	,0.789)(
810	,0.837)(
910	,0.872)(
1000	,0.885)
  } ;

\addplot coordinates {
(10	,0.002)(
110	,0.025)(
210	,0.067)(
310	,0.13)(
410	,0.131)(
510	,0.148)(
610	,0.159)(
710	,0.14)(
810	,0.121)(
910	,0.115)(
1000,	0.114
)

};

\addplot coordinates {
(10,	0)(
110,	0.039)(
210,	0.082)(
310	,0.111)(
410	,0.086)(
510	,0.075)(
610	,0.054)(
710	,0.039)(
810	,0.026)(
910	,0.025)(
1000,0.02)

};

\addlegendentry{(\ref{eqn:fin3}), $K=3$}
\addlegendentry{(\ref{eqn:fin3}), $K=4$}
\addlegendentry{(\ref{eqn:fin1}), (\ref{eqn:fin2}), (\ref{eqn:fin3}), $K=3$}
\addlegendentry{(\ref{eqn:fin1}), (\ref{eqn:fin2}), (\ref{eqn:fin3}), $K=4$}
\end{axis}
\end{tikzpicture}
\caption{Success probability of achievability conditions (\ref{eqn:fin3}) and success probability of conditions (\ref{eqn:fin1}), (\ref{eqn:fin2}), (\ref{eqn:fin3}) for all S-HK schemes except TIN with topology 1, $K=3$ and $K=4$}
\label{fig:topology1ach}
\end{figure}

In Fig. \ref{fig:topology1ach}, we plot the success probability of the achievability conditions (\ref{eqn:fin3}) alone and compare them with success probability of all conditions (\ref{eqn:fin1}), (\ref{eqn:fin2}), (\ref{eqn:fin3}) for all S-HK schemes except TIN for topology 1 with $K=3$ and $K=4$. It can be observed that the probability that achievability conditions are satisfied is much larger than the probability that all conditions are satisfied. It is worth noting that whenever the achievability conditions are satisfied, interference can be decoded and the resulting sum rate will be significantly better than the rate achieved by the TIN scheme. Therefore, even when the sum capacity conditions are not satisfied, there is significant improvement in the sum rate of S-HK schemes with interference decoding compare to the TIN scheme.  As $K$ increases, the probability of at least one interference signal being decodable increases as expected.

\section{Conclusions}
We obtained new sum capacity results for the general $K$-user Gaussian IC. We derived two sets of channel conditions under which S-HK schemes are sum capacity optimal for the $K$ user Gaussian IC. This general result also allows us to obtain all existing sum capacity results for 2-user GICs and partially connected GICs like the cascade, many-to-one and one-to-many GICs as special cases. We also study the probability that the channel conditions required for the sum capacity result are satisfied in random wireless networks using Monte Carlo simulations. The numerical results show that S-HK schemes are optimal with significant probability in the considered topologies. By selecting the best S-HK scheme for each channel condition, these results can be used for dynamic interference management and sum rate maximization in wireless networks.

\bibliographystyle{IEEEtran}
\bibliography{ptinref}


\appendices
\section{Converse part of proof of Theorem \ref{thm:covgen}}
\label{app:convthm1}
(Converse) Consider the genie-aided channel where each receiver $i \in [K]$ is given the  genie signal $\bs_i^n= \{\bx_i^n+\bn_i^n, \bx_j^n, j \in D(i)\}$, where $\bn_i^n \sim \mathcal{N}(\mathbf{0}, \sigma_i^2 \mathbf{I})$ and $E[n_i z_i]=\rho_i \sigma_i$, $0 <  \rho_i < 1$. For the result for TIN in \cite{shang2008}, the special case of this genie-aided channel where $D(i)$ is empty for all $i$ was used. Now, the sum capacity can be upper bounded as
\begin{align}
&nC_{sum}  \leq \sum_{i=1}^K I(\bx_i^n;\by_i^n,\bs_i^n) \nonumber \\
& = \sum_{i=1}^K I(\bx_i^n;\by_i^n,\bx_i^n+\bn_i^n | \bx_j^n, j \in D(i)) \nonumber\\
& = \sum_{i=1}^K \left[h(\bx_i^n+\bn_i^n)-h(\bn_i^n)\right]\nonumber\\
& ~+\sum_{i=1}^K h\left(\by_i^n| \bx_i^n+\bn_i^n , \bx_j^n, j \in D(i) \right) \nonumber\\
& ~-\sum_{i=1}^K \left[ h \left(\underset{j \in I(i)}{\sum} h_{ij}\bx_j^n +\bu_i^n\right) \right] \label{eqn:one}
\end{align}
where $u_i^n \sim \mathcal{N}(\mathbf{0}, (1-\rho_i^2) \mathbf{I})$, $\forall i \in [K]$.
Assuming
\begin{equation}
 1-\rho_i^2= \phi_i+ \underset{j \in I(i)}{\sum} h_{ij} \sigma_j^2,  \forall i \in [K]
 \label{epicondition}
 \end{equation}
 where $\phi_i \geq 0$, we can write
\begin{align*}
\mbox{cov}(\bu_i^n)& =  \mbox{cov}(\sqrt{\phi_i} \bn_0^n+\underset{j \in I(i)}{\sum} h_{ij} \bn_j^n),
\end{align*}
where $\bn_0^n \sim \mathcal{N}(\mathbf{0}, \mathbf{I})$ and is independent of $\bn_i^n, \forall i \in [K] $.
Now, we have
\begin{align}
& \mbox{exp} \left[ \frac{2}{n} h \left( \underset{j \in I(i)}{\sum} h_{ij}\bx_j^n +\bu_i^n\right)\right]  \nonumber\\
&=\mbox{exp} \left[ \frac{2}{n} h \left(\sqrt{\phi_i} \bn_0^n+ \underset{j \in I(i)}{\sum}(h_{ij} \bx_j^n+h_{ij} \bn_j^n) \right)\right]  \nonumber\\
& \overset{(a)}{\geq} \mbox{exp} \left[\frac{2}{n} h \left(\sqrt{\phi_i} \bn_0^n \right) \right] + \underset{j \in I(i)}{\sum} \mbox{exp} \left[\frac{2}{n} h \left( h_{ij} \bx_j^n+h_{ij} \bn_j^n\right)   \right]\nonumber\\
&= 2\pi e \phi_i+ \underset{j \in I(i)}{\sum} h_{ij}^2 \mbox{exp} \left[\frac{2}{n} h(\bx_j^n+ \bn_j^n) \right], \label{eqn:two}
 \end{align}
 where (a) follows from entropy-power inequality (EPI).
Therefore, we have
\begin{align*}
&\sum_{i=1}^K \left[h(\bx_i^n+\bn_i^n)-h\left(\underset{j \in I(i)}{\sum} h_{ij}\bx_j^n +\bu_i^n\right) \right] \\
&\leq \frac{n}{2} \sum_{i=1}^K \left[t_i- \log \left( 2\pi e \phi_i+ \underset{j \in I(i)}{\sum} h_{ij}^2 e^{t_j}\right) \right] \triangleq \frac{n}{2} f(\bt).
\end{align*}
where $t_i\triangleq \frac{2}{n} h(\bx_i^n+ \bn_i^n), \forall i \in [K]$ and $\bt$ is the vector of all $t_i$'s. From the power constraints, we have $t_i \leq \log \left[2\pi e (P_i+\sigma_i^2) \right], \ i\in [K].$ Under these constraints on $t_i$, it can be shown as in \cite{shang2008} that  $f(t)$ is maximized at  $t_k=\log \left[2\pi e (P_k+\sigma_k^2)\right]$ provided $\frac{\partial f}{\partial t_k}$ at $t_i=\log \left[2\pi e (P_i+\sigma_i^2)\right], \forall i \in [K] $  are greater than equal to 0. Thus, we have the conditions
\begin{align}
\sum_{i: k \in I(i)} \left[ \frac{h_{ik}^2}{1+Q_i-\rho_i^2}\right] & \leq \frac{1}{P_k+\sigma_k^2}, \ \forall k \in [K]. \label{eqn:conditionn11}
\end{align}
Therefore, we now have
\begin{align*}
C_{sum} & \leq \sum_{i=1}^K I(x_{iG};y_{iG},s_{iG})\\
 &=\sum_{i=1}^K I(x_{iG};y_{iG},x_{jG}, j \in D(i))\\
 &+\sum_{i=1}^K I(x_{iG};x_{iG}+n_{iG}|y_{iG},x_{jG}, j \in D(i))\\
 & \overset{(b)}{=}\sum_{i=1}^K I(x_{iG};y_{iG},x_{jG}, j \in D(i))\\
 &=\underset{i=1}{\overset{K}{\sum}} \frac{1}{2} \log \left[1+\frac{P_i}{1+Q_i}\right]
\end{align*}
where (b) is true if $I(x_{iG};x_{iG}+n_{i}|y_{iG},x_{jG}, j \in D(i))=0, \forall i$. From \cite[Lemma 8]{annvee09}, $I(x_{iG};x_{iG}+n_{iG}|y_{iG},x_{jG}, j \in D(i))=0$ iff
\begin{equation} \label{eqn:smart11}
\rho_i \sigma_i= 1+Q_i ,\forall i\in [K].
\end{equation}
Also, $\phi_i \geq 0, \ \forall i\in [K]$ which implies 
\begin{align}
 \underset{j \in I(i)}{\sum} \frac{h_{ij}^2(1+Q_j)^2}{\rho_j^2} & \leq  1-\rho_i^2,\label{eqn:conditionfin22}
\end{align}

Using the conditions (\ref{eqn:smart11}), (\ref{eqn:conditionfin22}) and (\ref{epicondition}), we get the conditions (\ref{eqn:fin1}) and (\ref{eqn:fin2}) for the converse.

\section{Proof of Theorem 2}
\label{app:proofthm3}
First, we consider a modified channel with no interference at receiver $k$. The sum capacity of the original channel is upper bounded by the sum capacity of the modified channel. 
Then, we derive a genie-aided upper bound for the modified channel using the 
genie signals $\bs_i^n$ at receiver $i$ for each $i \in [K]$ as follows:
\begin{align*}
\bs_i^n & = \{\bx_i^n+\bn_i^n, \bx_j^n, j \in \bar{G}(i)\} , \ \forall i\in [K] \backslash \{m,k\} \\
\bs_m^n & =\{\bx_j^n,\  j \in \bar{G}(m) \backslash k \}\\
\bs_k^n &=h_{mk}\bx_k^n+\underset{j \in G(m)}{\sum} h_{mj} \bx_j^n+ \bn_k^n
\end{align*}
where $\bar{G}(i)=[K] \backslash \{i,\{G(i)\}\}$, $n_i \sim \mathcal{N}(0, \sigma_i^2)$, $E[n_i z_i]=\rho_i \sigma_i$, and  $0 < \rho_i < 1$, for each $i \in [K] \backslash \{m\}$, and $\sigma_k=1$.
This choice of genie can be shown to be useful and smart under conditions (\ref{eqn:jointfin1})-(\ref{eqn:jointfin4}) to obtain the upper bound in the theorem statement. This proof generalizes the proof for the many-to-one GIC in \cite[Theorem 4]{GnaChaBha17} to the general $K$ user GIC.
Assume
\begin{align}
1-\rho_i^2 &= \phi_i+ \underset{j \in G(i)}{\sum} h_{ij}^2\sigma_j^2, i \in [K] \backslash \{m,k \}, \label{assm1}\\
1-\rho_k^2 &= \phi_k+ \underset{j \in G(m)}{\sum} h_{mj}^2 \sigma_j^2, \label{assm2}
\end{align}
where $\phi_i \geq 0$, and assume $h_{ij}^2 \leq 1, \,\, \forall i \in G(i), \ i \in [K]$. Now, we have
\begin{align*}
&nC_{sum} \leq \sum_{i=1}^K I(\bx_i^n;\by_i^n,\bs_i^n) \nonumber \\
& =  I(\bx_m^n;\by_m^n| \bs_m^n)  + I(\bx_k^n;\by_k^n,\bs_k^n) + \sum_{\underset{i \neq \{m,k\}}{i=1}}^K I(\bx_i^n;\by_i^n,\bs_i^n) \nonumber \\
& = h(\by_m^n |\bs_m^n)-h(\by_m^n |\bs_m^n, \bx_m^n )+h(\bs_k^n)+h(\by_k^n | \bs_k^n)\nonumber\\
& ~ ~ ~-h(\bz_k^n)-h(\bs_k^n| \by_k^n,\bx_k^n) + \sum_{\underset{i \neq \{m,k\}}{i=1}}^K I(\bx_i^n;\by_i^n,\bs_i^n)  \nonumber \\
&\overset{(a)}{=} h(\by_m^n |\bs_m^n)+h(\by_k^n | \bs_k^n)-h(\bz_k^n)\nonumber\\
&  -h\left(\sum_{j \in G(m)} h_{mj}\bx_j^n+\bu_k^n\right)+ \sum_{\underset{i \neq \{m,k\}}{i=1}}^K [h(\bx_i^n+\bn_i^n)] \nonumber \\
& +\sum_{\underset{i \neq \{m,k\}}{i=1}}^K  \left[h\left(\by_i^n| s_i^n\right) -h(\bn_i^n)  -h\left(\underset{j \in G(i)}{\sum} h_{ij}\bx_j^n +\bu_i^n\right)  \right]  
\end{align*}
where (a) follows because $h(\bs_k^n)=h(\by_m^n |\bs_m^n, \bx_m^n )$, $\bu_i^n \sim \mathcal{N}(\mathbf{0}, (1-\rho_i^2) \mathbf{I})$ , $\forall i \in [K]\backslash \{m\}$. 
Also note that from (\ref{assm1}) and (\ref{assm2}), we have
\begin{align*}
\mbox{cov}(\bu_i^n)& = \mbox{cov}(\sqrt{\phi_i} \bn_0^n+\underset{j \in G(i)}{\sum} h_{ij} \bn_j^n) \forall i \in [K] \backslash \{m,k\},\\
 \mbox{cov}(\bu_k^n)& = \mbox{cov}(\sqrt{\phi_k} \bn_0^n+\underset{j \in G(m)}{\sum} h_{mj} \bn_j^n),\\
\end{align*}
 where $\bn_0^n \sim \mathcal{N}(\mathbf{0}, \mathbf{I})$ and $\bn_0^n $ is independent of $\bn_i^n, \ i \in [K] \backslash \{m\} $. 
From EPI, we have
 \begin{align}
& \exp \left[ \frac{2}{n} h \left( \underset{j \in G(i)}{\sum} h_{ij}\bx_j^n +\bu_i^n\right)\right]  \nonumber\\
&=\exp \left[ \frac{2}{n} h \left(\sqrt{\phi_i} \bn_0^n+ \underset{j \in G(i)}{\sum}(h_{ij} \bx_j^n+h_{ij} \bn_j^n) \right)\right]  \nonumber\\
& \geq \exp \left[\frac{2}{n} h \left(\sqrt{\phi_i} \bn_0^n \right) \right] + \underset{j \in G(i)}{\sum} \exp \left[\frac{2}{n} h \left( h_{ij} \bx_j^n+h_{ij} \bn_j^n\right) \right] \nonumber\\
&= 2\pi e \phi_i+ \underset{j \in G(i)}{\sum}  \exp \left[\frac{2}{n} h(\bx_j^n+ \bn_j^n) \right] \exp\left[\frac{2}{n}\log (h_{ij}^n)\right]\\
&= 2\pi e \phi_i+ \underset{j \in G(i)}{\sum} h_{ij}^2 \exp \left[\frac{2}{n} h(\bx_j^n+ \bn_j^n) \right] \label{eqn:two}
 \end{align}
Considering the terms that are not directly maximized by Gaussian inputs, we get
\begin{align*}
&\sum_{\underset{i \neq m,k}{i=1}}^K \left[h(\bx_i^n+\bn_i^n)-h\left(\underset{j \in G(i)}{\sum} h_{ij}\bx_j^n +\bu_i^n\right)\right]\\
& ~ - h\left(\underset{j \in G(m)}{\sum} h_{mj}\bx_j^n +\bu_k^n\right)  \\
&\overset{(b)}{\leq} \frac{n}{2} \sum_{\underset{i \neq m,k}{i=1}}^K \left[t_i- \log \left( 2\pi e \phi_i+ \underset{j \in G(i)}{\sum} h_{ij}^2 e^{t_j}\right) \right]\\
&~ - \frac{n}{2} \log \left( 2\pi e \phi_k+ \underset{j \in G(m)}{\sum} h_{mj}^2 e^{t_j}\right) \triangleq \frac{n}{2} f(\bt),
\end{align*}
where (b) follows from (\ref{eqn:two}) and $t_i\triangleq \frac{2}{n} h(\bx_i^n+\bn_i^n), \ i \in [K] $. 
From the power constraints, we have
\[ t_i \leq \log \left[2\pi e (P_i+\sigma_i^2) \right], \ i=1,\cdots,K .\]\\
Since $m, k \notin G(i)$, $\forall i \in [K] \backslash \{k\}$, terms $t_m$, $t_k$ do not appear in $f(t)$, and we consider the optimization problem
\begin{equation}
\begin{aligned}
& \text{max}
& &   f(\bt)\\
& \text{s.t}
& & t_i \leq \log \left[2\pi e (P_i+\sigma_i^2)\right], \nonumber \\ 
& & & \forall i\in [K] \backslash \{m,k\}.
\end{aligned}
\end{equation}
The Lagrangian for the above optimization problem is 
\[ L=f(\bt)- \underset{i\neq m,k}{\sum} \mu_i[t_i- \log \left[2\pi e (P_i+\sigma_i^2)\right]], \]
where $\mu_i \geq 0 , \forall i \in [K]$. At optimal $t_r, \frac{\partial L}{\partial t_r}=0$. We want optimal $t_r=\frac{1}{2}\log \left[2\pi e (P_r+\sigma_r^2)\right]$ for Gaussian inputs to be optimal for the genie-aided channel. From KKT conditions, $\frac{\partial f}{\partial t_r} \geq 0$ at $t_r= \log\left[2\pi e (P_r+\sigma_r^2)\right], \forall r \in [K] \backslash \{m,k \} $ (since $ \mu_r \geq 0$). Thus, we have
\begin{IEEEeqnarray}{lcr}
\frac{1}{P_r+\sigma_r^2}   \geq  \sum_{\underset{i \neq \{m,k\}}{i: r \in G(i)}} \left[ \frac{h_{ir}^2}{1+Q_i-\rho_i^2}\right]\nonumber \\
+ \delta_r \left[\frac{h_{mr}^2}{1+Q_m-\rho_k^2} \right] , \forall \ r \in [K] \backslash  \{m,k\},
 \label{eqn:conditionn1}
\end{IEEEeqnarray}
where $Q_i, \delta_r$ are defined as in the theorem statement.
 Therefore, we now have
\begin{align*}
C_{sum} & \leq \sum_{i=1}^K I(x_{iG};y_{iG},s_{iG})\\
&= I(x_{mG};y_{mG}| s_{mG})+I(x_{kG};y_{kG}, s_{kG})\\
& +\sum_{\underset{i \neq \{m,k\}}{i=1}}^K I(x_{iG};y_{iG},s_{iG})\\
& \overset{(c)}{=} I(x_{mG};y_{mG}| s_{mG})+I(x_{kG};s_{kG})\\
& +\sum_{\underset{i \neq \{m,k\}}{i=1}}^K I(x_{iG};y_{iG} | x_{jG}, j \in \bar{G}(i))\\
& \overset{(d)}{=} I(x_{mG},x_{kG};y_{mG}| s_{mG})\\
&+\sum_{\underset{i \neq \{m,k\}}{i=1}}^K I(x_{iG};y_{iG} | x_{jG}, j \in \bar{G}(i)),
\end{align*}
(c) is valid when $I(x_{iG};x_{iG}+n_i|y_{iG},x_{jG},j \in \bar{G}(i))=0$, $\forall i\in [K] \backslash \{m,k\}$ and $I(x_{kG};y_{kG}|s_{kG})=0$. From \cite[Lemma 8]{annvee09}, $I(x_{iG};x_{iG}+n_i|y_{iG},x_{jG},j \in \bar{G}(i))=I \left(x_{iG};x_{iG}+n_i \bigg |(x_{iG}+ \underset{j \in I(i)}{\sum} h_{ij}x_{jG}+z_i)\right)=0$, $\forall i\in [K] \backslash \{m,k\}$ and $I(x_{kG};y_{kG} |s_{kG})=I \left(x_{kG};x_{kG}+z_k \bigg | \frac{x_{kG}+\underset{j \in I(m)}{\sum} h_{mj}x_{jG}+n_k}{h_{mk}} \right)=0$ iff 
\begin{align} 
\rho_i \sigma_i & = 1+Q_i ,\forall i\in [K] \backslash \{m,k\}, \label{eqn:smart1} \\
\rho_k h_{mk} &= 1+Q_m. \label{eqn:smart2}
\end{align}
(d) is valid since genie $s_k$ is chosen such that $h(s_k)=h(y_m|s_m,x_m)$ which implies $I(x_{kG};s_{kG})=I(x_{kG};y_{mG}|s_{mG},x_{mG})$.

Using $\phi_i \geq 0$ and (\ref{eqn:smart1}), we get conditions (\ref{eqn:jointfin1})-(\ref{eqn:jointfin4}) for the upper bound on  $C_{sum}$ to be valid.

\section{2-user GIC results}
\label{app:2user}
For a 2 user GIC, there are only 4 possible S-HK schemes.
For each S-HK scheme, the sum capacity and the required conditions using Theorem \ref{thm:covgen} are given in Table \ref{table:shk2user1}.
\begin{table*}
\centering

\begin{tabular}{|p{1.38cm} | p{2.35cm} | p{2.7cm}|p{1cm}|} 
  \hline
  S-HK scheme & Optimality conditions & Sum capacity & Matches\\ \hline
  $I(1)=\{2\}$  & $|h_{12}(1+h_{21}^2 P_1)|+$  & $\frac{1}{2} \log(1+\frac{P_1}{1+h_{12}^2 P_1})$  & \cite{shakrache09,motkha09,annvee09}\\ 
  $I(2)=\{1\}$ & $|h_{21}(1+h_{12}^2 P_2)| \leq 1$ &  $+\frac{1}{2}\log(1+ \frac{P_2}{1+h_{21}^2 P_1})$ & \\
 \hline
 $I(1)=\{\}$ & $h_{21}^2 \leq 1$, & $\frac{1}{2} \log(1+P_1)+$ & Thm. 10 in \cite{motkha09}\\
 $I(2)=\{1\}$ & $h_{12}^2 \geq \frac{1+P_1}{1+h_{21}^2 P_1}$  &  $\frac{1}{2}\log(1+ \frac{P_2}{1+h_{21}^2 P_1})$ &\\
 \hline
 $I(1)=\{2\}$ & $h_{12}^2 \leq 1 $, & $\frac{1}{2} \log(1+P_2)+$ & Thm. 10 in \cite{motkha09}\\
 $I(2)=\{\}$ &  $h_{21}^2 \geq \frac{1+P_2}{1+h_{12}^2 P_2}$ & $\frac{1}{2} \log(1+ \frac{P_1}{1+h_{12}^2 P_2})$ &\\ 
 \hline
 $I(1)=\{\}$ & $h_{12}^2 \geq 1+P_1 $, & $\frac{1}{2} \log(1+P_1)+$ & \cite{car75}\\
 $I(2)=\{\}$ & $h_{21}^2 \geq 1+P_2$ & $\frac{1}{2} \log(1+P_2)$ & \\
 \hline
\end{tabular}

\caption{Optimality conditions and sum capacity of S-HK schemes for a 2 user GIC using Theorem \ref{thm:covgen}.}
\label{table:shk2user1}
\end{table*}
For each S-HK scheme, the sum capacity and the required conditions using Theorems \ref{thm:achGIC} and \ref{thm:convjoint} are given in Table \ref{table:shk2user2} and these conditions are plotted in Fig \ref{fig:2-user}.

\begin{figure}
\centering
\resizebox{7cm}{6.5cm}{
\begin{tikzpicture}
\begin{axis}[xlabel={$h_{12}$},
    ylabel={$h_{21}$},
    xmin=0, xmax=3,
    ymin=0, ymax=3,
    xtick={0,0.5,1,1.5,2,2.5,3},
    ytick={0,0.5,1,1.5,2,2.5,3},
]
\addplot[ fill=blue, 
                    fill opacity=0.2, color=blue]
    coordinates {	(1,0)	(0.998,0.001)	(0.996,0.002)	(0.994,0.003)	(0.992,0.004)	(0.99,0.005)	(0.988,0.006)	(0.986,0.007)	(0.984,0.008)	(0.982,0.009)	(0.98,0.01)	(0.978,0.011)	(0.976,0.012)	(0.963,0.019)	(0.961,0.02)	(0.948,0.027)	(0.937,0.033)	(0.926,0.039)	(0.917,0.044)	(0.91,0.048)	(0.903,0.052)	(0.896,0.056)	(0.891,0.059)	(0.884,0.063)	(0.879,0.066)	(0.874,0.069)	(0.869,0.072)	(0.864,0.075)	(0.859,0.078)	(0.851,0.083)	(0.846,0.086)	(0.838,0.091)	(0.827,0.098)	(0.816,0.105)	(0.805,0.112)	(0.802,0.114)	(0.799,0.116)	(0.796,0.118)	(0.793,0.12)	(0.79,0.122)	(0.787,0.124)	(0.784,0.126)	(0.781,0.128)	(0.778,0.13)	(0.775,0.132)	(0.772,0.134)	(0.769,0.136)	(0.759,0.143)	(0.756,0.145)	(0.746,0.152)	(0.739,0.157)	(0.729,0.164)	(0.725,0.167)	(0.718,0.172)	(0.711,0.177)	(0.707,0.18)	(0.703,0.183)	(0.696,0.188)	(0.692,0.191)	(0.688,0.194)	(0.684,0.197)	(0.68,0.2)	(0.671,0.207)	(0.667,0.21)	(0.658,0.217)	(0.649,0.224)	(0.64,0.231)	(0.635,0.235)	(0.63,0.239)	(0.625,0.243)	(0.62,0.247)	(0.614,0.252)	(0.609,0.256)	(0.603,0.261)	(0.598,0.265)	(0.592,0.27)	(0.586,0.275)	(0.579,0.281)	(0.573,0.286)	(0.566,0.292)	(0.558,0.299)	(0.551,0.305)	(0.543,0.312)	(0.534,0.32)	(0.525,0.328)	(0.515,0.337)	(0.503,0.348)	(0.49,0.36)	(0.477,0.372)	(0.476,0.373)	(0.475,0.374)	(0.474,0.375)	(0.473,0.376)	(0.46,0.388)	(0.459,0.389)	(0.458,0.39)	(0.457,0.391)	(0.456,0.392)	(0.455,0.393)	(0.454,0.394)	(0.453,0.395)	(0.452,0.396)	(0.451,0.397)	(0.45,0.398)	(0.449,0.399)	(0.448,0.4)	(0.435,0.412)	(0.434,0.413)	(0.433,0.414)	(0.432,0.415)	(0.431,0.416)	(0.43,0.417)	(0.429,0.418)	(0.428,0.419)	(0.427,0.42)	(0.426,0.421)	(0.425,0.422)	(0.424,0.423)	(0.423,0.424)	(0.422,0.425)	(0.421,0.426)	(0.42,0.427)	(0.419,0.428)	(0.418,0.429)	(0.417,0.43)	(0.416,0.431)	(0.415,0.432)	(0.414,0.433)	(0.413,0.434)	(0.412,0.435)	(0.4,0.448)	(0.399,0.449)	(0.398,0.45)	(0.397,0.451)	(0.396,0.452)	(0.395,0.453)	(0.394,0.454)	(0.393,0.455)	(0.392,0.456)	(0.391,0.457)	(0.39,0.458)	(0.389,0.459)	(0.388,0.46)	(0.376,0.473)	(0.375,0.474)	(0.374,0.475)	(0.373,0.476)	(0.372,0.477)	(0.36,0.49)	(0.348,0.503)	(0.337,0.515)	(0.328,0.525)	(0.32,0.534)	(0.312,0.543)	(0.305,0.551)	(0.299,0.558)	(0.292,0.566)	(0.286,0.573)	(0.281,0.579)	(0.275,0.586)	(0.27,0.592)	(0.265,0.598)	(0.261,0.603)	(0.256,0.609)	(0.252,0.614)	(0.247,0.62)	(0.243,0.625)	(0.239,0.63)	(0.235,0.635)	(0.231,0.64)	(0.224,0.649)	(0.217,0.658)	(0.21,0.667)	(0.207,0.671)	(0.2,0.68)	(0.197,0.684)	(0.194,0.688)	(0.191,0.692)	(0.188,0.696)	(0.183,0.703)	(0.18,0.707)	(0.177,0.711)	(0.172,0.718)	(0.167,0.725)	(0.164,0.729)	(0.157,0.739)	(0.152,0.746)	(0.145,0.756)	(0.143,0.759)	(0.136,0.769)	(0.134,0.772)	(0.132,0.775)	(0.13,0.778)	(0.128,0.781)	(0.126,0.784)	(0.124,0.787)	(0.122,0.79)	(0.12,0.793)	(0.118,0.796)	(0.116,0.799)	(0.114,0.802)	(0.112,0.805)	(0.105,0.816)	(0.098,0.827)	(0.091,0.838)	(0.086,0.846)	(0.083,0.851)	(0.078,0.859)	(0.075,0.864)	(0.072,0.869)	(0.069,0.874)	(0.066,0.879)	(0.063,0.884)	(0.059,0.891)	(0.056,0.896)	(0.052,0.903)	(0.048,0.91)	(0.044,0.917)	(0.039,0.926)	(0.033,0.937)	(0.027,0.948)	(0.02,0.961)	(0.019,0.963)	(0.012,0.976)	(0.011,0.978)	(0.01,0.98)	(0.009,0.982)	(0.008,0.984)	(0.007,0.986)	(0.006,0.988)	(0.005,0.99)	(0.004,0.992)	(0.003,0.994)	(0.002,0.996)	(0.001,0.998)	(0,0)	(0,0)	(1,0)	
    
     };
  \addplot[ fill=red, 
                    fill opacity=0.2, color=red]
    coordinates {	(1,1)	(1.007,0.987)	(1.008,0.985)	(1.009,0.983)	(1.01,0.981)	(1.011,0.979)	(1.012,0.977)	(1.013,0.975)	(1.014,0.973)	(1.015,0.971)	(1.016,0.969)	(1.017,0.967)	(1.018,0.965)	(1.019,0.963)	(1.02,0.961)	(1.021,0.959)	(1.022,0.957)	(1.023,0.955)	(1.024,0.953)	(1.025,0.951)	(1.026,0.949)	(1.027,0.947)	(1.028,0.945)	(1.029,0.943)	(1.03,0.941)	(1.031,0.939)	(1.032,0.937)	(1.039,0.924)	(1.04,0.922)	(1.041,0.92)	(1.042,0.918)	(1.043,0.916)	(1.044,0.914)	(1.045,0.912)	(1.046,0.91)	(1.053,0.897)	(1.054,0.895)	(1.055,0.893)	(1.056,0.891)	(1.057,0.889)	(1.058,0.887)	(1.065,0.874)	(1.066,0.872)	(1.067,0.87)	(1.074,0.857)	(1.075,0.855)	(1.076,0.853)	(1.083,0.84)	(1.084,0.838)	(1.091,0.825)	(1.092,0.823)	(1.099,0.81)	(1.106,0.797)	(1.113,0.784)	(1.12,0.771)	(1.126,0.76)	(1.132,0.749)	(1.139,0.736)	(1.145,0.725)	(1.151,0.714)	(1.157,0.703)	(1.163,0.692)	(1.17,0.679)	(1.176,0.668)	(1.182,0.657)	(1.189,0.644)	(1.195,0.633)	(1.202,0.62)	(1.209,0.607)	(1.216,0.594)	(1.217,0.592)	(1.224,0.579)	(1.225,0.577)	(1.226,0.575)	(1.233,0.562)	(1.234,0.56)	(1.235,0.558)	(1.242,0.545)	(1.243,0.543)	(1.244,0.541)	(1.245,0.539)	(1.246,0.537)	(1.247,0.535)	(1.254,0.522)	(1.255,0.52)	(1.256,0.518)	(1.257,0.516)	(1.258,0.514)	(1.259,0.512)	(1.26,0.51)	(1.261,0.508)	(1.262,0.506)	(1.263,0.504)	(1.264,0.502)	(1.265,0.5)	(1.266,0.498)	(1.267,0.496)	(1.268,0.494)	(1.269,0.492)	(1.27,0.49)	(1.271,0.488)	(1.272,0.486)	(1.273,0.484)	(1.274,0.482)	(1.275,0.48)	(1.276,0.478)	(1.277,0.476)	(1.278,0.474)	(1.279,0.472)	(1.28,0.47)	(1.281,0.468)	(1.282,0.466)	(1.283,0.464)	(1.284,0.462)	(1.285,0.46)	(1.286,0.458)	(1.287,0.456)	(1.288,0.454)	(1.294,0.441)	(1.295,0.439)	(1.296,0.437)	(1.297,0.435)	(1.303,0.422)	(1.309,0.409)	(1.314,0.398)	(1.319,0.387)	(1.322,0.38)	(1.326,0.371)	(1.329,0.364)	(1.332,0.357)	(1.334,0.352)	(1.337,0.345)	(1.339,0.34)	(1.341,0.335)	(1.343,0.33)	(1.345,0.325)	(1.347,0.32)	(1.349,0.315)	(1.352,0.307)	(1.355,0.299)	(1.359,0.288)	(1.363,0.277)	(1.364,0.274)	(1.365,0.271)	(1.372,0.25)	(1.373,0.247)	(1.374,0.244)	(1.378,0.231)	(1.381,0.221)	(1.383,0.214)	(1.386,0.203)	(1.389,0.192)	(1.39,0.188)	(1.391,0.184)	(1.392,0.18)	(1.394,0.171)	(1.395,0.167)	(1.397,0.158)	(1.398,0.153)	(1.399,0.148)	(1.4,0.143)	(1.401,0.138)	(1.403,0.127)	(1.404,0.121)	(1.405,0.115)	(1.406,0.109)	(1.407,0.102)	(1.408,0.095)	(1.409,0.087)	(1.41,0.078)	(1.411,0.068)	(1.412,0.057)	(1.413,0.042)	(1.414,0.021)	(1.414,0.02)	(1.414,0.019)	(1.414,0.018)	(1.415,0)	(3,0)	(3,1)	(1,1)																																																										
    };
     \addplot[ fill=brown, 
                    fill opacity=0.2, color=brown]
    coordinates {	(0,1.415)	(0.018,1.414)	(0.019,1.414)	(0.02,1.414)	(0.042,1.413)	(0.057,1.412)	(0.068,1.411)	(0.078,1.41)	(0.087,1.409)	(0.095,1.408)	(0.102,1.407)	(0.109,1.406)	(0.115,1.405)	(0.121,1.404)	(0.127,1.403)	(0.138,1.401)	(0.143,1.4)	(0.148,1.399)	(0.153,1.398)	(0.158,1.397)	(0.167,1.395)	(0.171,1.394)	(0.18,1.392)	(0.184,1.391)	(0.188,1.39)	(0.192,1.389)	(0.203,1.386)	(0.214,1.383)	(0.221,1.381)	(0.231,1.378)	(0.244,1.374)	(0.247,1.373)	(0.25,1.372)	(0.271,1.365)	(0.274,1.364)	(0.277,1.363)	(0.288,1.359)	(0.299,1.355)	(0.307,1.352)	(0.315,1.349)	(0.32,1.347)	(0.325,1.345)	(0.33,1.343)	(0.335,1.341)	(0.34,1.339)	(0.345,1.337)	(0.352,1.334)	(0.357,1.332)	(0.364,1.329)	(0.371,1.326)	(0.38,1.322)	(0.387,1.319)	(0.398,1.314)	(0.409,1.309)	(0.422,1.303)	(0.435,1.297)	(0.437,1.296)	(0.439,1.295)	(0.441,1.294)	(0.454,1.288)	(0.456,1.287)	(0.458,1.286)	(0.46,1.285)	(0.462,1.284)	(0.464,1.283)	(0.466,1.282)	(0.468,1.281)	(0.47,1.28)	(0.472,1.279)	(0.474,1.278)	(0.476,1.277)	(0.478,1.276)	(0.48,1.275)	(0.482,1.274)	(0.484,1.273)	(0.486,1.272)	(0.488,1.271)	(0.49,1.27)	(0.492,1.269)	(0.494,1.268)	(0.496,1.267)	(0.498,1.266)	(0.5,1.265)	(0.502,1.264)	(0.504,1.263)	(0.506,1.262)	(0.508,1.261)	(0.51,1.26)	(0.512,1.259)	(0.514,1.258)	(0.516,1.257)	(0.518,1.256)	(0.52,1.255)	(0.522,1.254)	(0.535,1.247)	(0.537,1.246)	(0.539,1.245)	(0.541,1.244)	(0.543,1.243)	(0.545,1.242)	(0.558,1.235)	(0.56,1.234)	(0.562,1.233)	(0.575,1.226)	(0.577,1.225)	(0.579,1.224)	(0.592,1.217)	(0.594,1.216)	(0.607,1.209)	(0.62,1.202)	(0.633,1.195)	(0.644,1.189)	(0.657,1.182)	(0.668,1.176)	(0.679,1.17)	(0.692,1.163)	(0.703,1.157)	(0.714,1.151)	(0.725,1.145)	(0.736,1.139)	(0.749,1.132)	(0.76,1.126)	(0.771,1.12)	(0.784,1.113)	(0.797,1.106)	(0.81,1.099)	(0.823,1.092)	(0.825,1.091)	(0.838,1.084)	(0.84,1.083)	(0.853,1.076)	(0.855,1.075)	(0.857,1.074)	(0.87,1.067)	(0.872,1.066)	(0.874,1.065)	(0.887,1.058)	(0.889,1.057)	(0.891,1.056)	(0.893,1.055)	(0.895,1.054)	(0.897,1.053)	(0.91,1.046)	(0.912,1.045)	(0.914,1.044)	(0.916,1.043)	(0.918,1.042)	(0.92,1.041)	(0.922,1.04)	(0.924,1.039)	(0.937,1.032)	(0.939,1.031)	(0.941,1.03)	(0.943,1.029)	(0.945,1.028)	(0.947,1.027)	(0.949,1.026)	(0.951,1.025)	(0.953,1.024)	(0.955,1.023)	(0.957,1.022)	(0.959,1.021)	(0.961,1.02)	(0.963,1.019)	(0.965,1.018)	(0.967,1.017)	(0.969,1.016)	(0.971,1.015)	(0.973,1.014)	(0.975,1.013)	(0.977,1.012)	(0.979,1.011)	(0.981,1.01)	(0.983,1.009)	(0.985,1.008)	(0.987,1.007)	(1,1)	(1,3)	(0,3)	(0,1.415)	};				
    \addplot[ fill=black, 
                    fill opacity=0.2, color=black]
    coordinates {(1.415,1.415)	(3,1.415)	(3,3)	(1.415,3)	(1.415,1.415)};
     \addplot[ fill=orange, 
                    fill opacity=0.2, color=orange]
    coordinates {	(1.414,0)	(1.414,0.001)	(1.414,0.002)	(1.414,0.003)	(1.414,0.004)	(1.414,0.005)	(1.414,0.006)	(1.414,0.007)	(1.414,0.008)	(1.414,0.009)	(1.414,0.01)	(1.414,0.011)	(1.414,0.012)	(1.414,0.013)	(1.414,0.014)	(1.414,0.015)	(1.414,0.016)	(1.414,0.017)	(1.413,0.037)	(1.413,0.038)	(1.413,0.039)	(1.413,0.04)	(1.413,0.041)	(1.412,0.056)	(1.411,0.067)	(1.41,0.077)	(1.409,0.086)	(1.408,0.094)	(1.407,0.101)	(1.406,0.108)	(1.405,0.114)	(1.404,0.12)	(1.403,0.126)	(1.402,0.132)	(1.401,0.137)	(1.4,0.142)	(1.399,0.147)	(1.398,0.152)	(1.397,0.157)	(1.396,0.162)	(1.395,0.166)	(1.393,0.175)	(1.392,0.179)	(1.391,0.183)	(1.39,0.187)	(1.389,0.191)	(1.388,0.195)	(1.387,0.199)	(1.384,0.21)	(1.382,0.217)	(1.38,0.224)	(1.377,0.234)	(1.376,0.237)	(1.375,0.24)	(1.374,0.243)	(1.371,0.253)	(1.37,0.256)	(1.369,0.259)	(1.368,0.262)	(1.367,0.265)	(1.366,0.268)	(1.362,0.279)	(1.361,0.282)	(1.36,0.285)	(1.356,0.296)	(1.353,0.304)	(1.35,0.312)	(1.348,0.317)	(1.346,0.322)	(1.344,0.327)	(1.342,0.332)	(1.34,0.337)	(1.338,0.342)	(1.336,0.347)	(1.333,0.354)	(1.331,0.359)	(1.328,0.366)	(1.325,0.373)	(1.321,0.382)	(1.318,0.389)	(1.313,0.4)	(1.308,0.411)	(1.307,0.413)	(1.302,0.424)	(1.301,0.426)	(1.3,0.428)	(1.299,0.43)	(1.298,0.432)	(1.293,0.443)	(1.292,0.445)	(1.291,0.447)	(1.29,0.449)	(1.289,0.451)	(1.288,0.453)	(1.287,0.455)	(1.286,0.457)	(1.285,0.459)	(1.284,0.461)	(1.283,0.463)	(1.282,0.465)	(1.281,0.467)	(1.28,0.469)	(1.279,0.471)	(1.278,0.473)	(1.277,0.475)	(1.276,0.477)	(1.275,0.479)	(1.274,0.481)	(1.273,0.483)	(1.272,0.485)	(1.271,0.487)	(1.27,0.489)	(1.269,0.491)	(1.268,0.493)	(1.267,0.495)	(1.266,0.497)	(1.265,0.499)	(1.264,0.501)	(1.263,0.503)	(1.262,0.505)	(1.261,0.507)	(1.26,0.509)	(1.259,0.511)	(1.258,0.513)	(1.257,0.515)	(1.256,0.517)	(1.255,0.519)	(1.254,0.521)	(1.253,0.523)	(1.252,0.525)	(1.251,0.527)	(1.25,0.529)	(1.249,0.531)	(1.248,0.533)	(1.242,0.544)	(1.241,0.546)	(1.24,0.548)	(1.239,0.55)	(1.238,0.552)	(1.237,0.554)	(1.236,0.556)	(1.23,0.567)	(1.229,0.569)	(1.228,0.571)	(1.227,0.573)	(1.221,0.584)	(1.22,0.586)	(1.219,0.588)	(1.218,0.59)	(1.212,0.601)	(1.211,0.603)	(1.21,0.605)	(1.204,0.616)	(1.203,0.618)	(1.197,0.629)	(1.196,0.631)	(1.19,0.642)	(1.184,0.653)	(1.183,0.655)	(1.177,0.666)	(1.171,0.677)	(1.165,0.688)	(1.164,0.69)	(1.158,0.701)	(1.152,0.712)	(1.146,0.723)	(1.14,0.734)	(1.134,0.745)	(1.133,0.747)	(1.127,0.758)	(1.121,0.769)	(1.115,0.78)	(1.114,0.782)	(1.108,0.793)	(1.107,0.795)	(1.101,0.806)	(1.1,0.808)	(1.094,0.819)	(1.093,0.821)	(1.087,0.832)	(1.086,0.834)	(1.085,0.836)	(1.079,0.847)	(1.078,0.849)	(1.077,0.851)	(1.071,0.862)	(1.07,0.864)	(1.069,0.866)	(1.068,0.868)	(1.062,0.879)	(1.061,0.881)	(1.06,0.883)	(1.059,0.885)	(1.053,0.896)	(1.052,0.898)	(1.051,0.9)	(1.05,0.902)	(1.049,0.904)	(1.048,0.906)	(1.047,0.908)	(1.041,0.919)	(1.04,0.921)	(1.039,0.923)	(1.038,0.925)	(1.037,0.927)	(1.036,0.929)	(1.035,0.931)	(1.034,0.933)	(1.033,0.935)	(1.027,0.946)	(1.026,0.948)	(1.025,0.95)	(1.024,0.952)	(1.023,0.954)	(1.022,0.956)	(1.021,0.958)	(1.02,0.96)	(1.019,0.962)	(1.018,0.964)	(1.017,0.966)	(1.016,0.968)	(1.015,0.97)	(1.014,0.972)	(1.013,0.974)	(1.012,0.976)	(1.011,0.978)	(1.01,0.98)	(1.009,0.982)	(1.008,0.984)	(1.007,0.986)	(1.006,0.988)	(1.005,0.99)	(1.004,0.992)	(1.003,0.994)	(1.002,0.996)	(1.001,0.998)	(1,1)	(1,0)	(1.414,0)};
     \addplot[ fill=green, 
                    fill opacity=0.2, color=green]
    coordinates {	(1,1)	(0.998,1.001)	(0.996,1.002)	(0.994,1.003)	(0.992,1.004)	(0.99,1.005)	(0.988,1.006)	(0.986,1.007)	(0.984,1.008)	(0.982,1.009)	(0.98,1.01)	(0.978,1.011)	(0.976,1.012)	(0.974,1.013)	(0.972,1.014)	(0.97,1.015)	(0.968,1.016)	(0.966,1.017)	(0.964,1.018)	(0.962,1.019)	(0.96,1.02)	(0.958,1.021)	(0.956,1.022)	(0.954,1.023)	(0.952,1.024)	(0.95,1.025)	(0.948,1.026)	(0.946,1.027)	(0.935,1.033)	(0.933,1.034)	(0.931,1.035)	(0.929,1.036)	(0.927,1.037)	(0.925,1.038)	(0.923,1.039)	(0.921,1.04)	(0.908,1.047)	(0.906,1.048)	(0.904,1.049)	(0.902,1.05)	(0.9,1.051)	(0.898,1.052)	(0.885,1.059)	(0.883,1.06)	(0.881,1.061)	(0.879,1.062)	(0.868,1.068)	(0.866,1.069)	(0.864,1.07)	(0.851,1.077)	(0.849,1.078)	(0.836,1.085)	(0.834,1.086)	(0.832,1.087)	(0.821,1.093)	(0.819,1.094)	(0.808,1.1)	(0.795,1.107)	(0.793,1.108)	(0.782,1.114)	(0.769,1.121)	(0.758,1.127)	(0.747,1.133)	(0.745,1.134)	(0.734,1.14)	(0.723,1.146)	(0.712,1.152)	(0.701,1.158)	(0.69,1.164)	(0.688,1.165)	(0.677,1.171)	(0.666,1.177)	(0.655,1.183)	(0.642,1.19)	(0.631,1.196)	(0.629,1.197)	(0.618,1.203)	(0.605,1.21)	(0.603,1.211)	(0.59,1.218)	(0.588,1.219)	(0.586,1.22)	(0.584,1.221)	(0.573,1.227)	(0.571,1.228)	(0.569,1.229)	(0.556,1.236)	(0.554,1.237)	(0.552,1.238)	(0.55,1.239)	(0.548,1.24)	(0.546,1.241)	(0.533,1.248)	(0.531,1.249)	(0.529,1.25)	(0.527,1.251)	(0.525,1.252)	(0.523,1.253)	(0.521,1.254)	(0.519,1.255)	(0.517,1.256)	(0.515,1.257)	(0.513,1.258)	(0.511,1.259)	(0.509,1.26)	(0.507,1.261)	(0.505,1.262)	(0.503,1.263)	(0.501,1.264)	(0.499,1.265)	(0.497,1.266)	(0.495,1.267)	(0.493,1.268)	(0.491,1.269)	(0.489,1.27)	(0.487,1.271)	(0.485,1.272)	(0.483,1.273)	(0.481,1.274)	(0.479,1.275)	(0.477,1.276)	(0.475,1.277)	(0.473,1.278)	(0.471,1.279)	(0.469,1.28)	(0.467,1.281)	(0.465,1.282)	(0.463,1.283)	(0.461,1.284)	(0.459,1.285)	(0.457,1.286)	(0.455,1.287)	(0.453,1.288)	(0.451,1.289)	(0.449,1.29)	(0.447,1.291)	(0.445,1.292)	(0.443,1.293)	(0.43,1.299)	(0.428,1.3)	(0.426,1.301)	(0.424,1.302)	(0.411,1.308)	(0.4,1.313)	(0.389,1.318)	(0.382,1.321)	(0.373,1.325)	(0.366,1.328)	(0.359,1.331)	(0.354,1.333)	(0.347,1.336)	(0.342,1.338)	(0.337,1.34)	(0.332,1.342)	(0.327,1.344)	(0.322,1.346)	(0.317,1.348)	(0.312,1.35)	(0.304,1.353)	(0.296,1.356)	(0.285,1.36)	(0.282,1.361)	(0.279,1.362)	(0.268,1.366)	(0.265,1.367)	(0.262,1.368)	(0.259,1.369)	(0.256,1.37)	(0.253,1.371)	(0.243,1.374)	(0.24,1.375)	(0.237,1.376)	(0.234,1.377)	(0.224,1.38)	(0.217,1.382)	(0.21,1.384)	(0.199,1.387)	(0.195,1.388)	(0.191,1.389)	(0.187,1.39)	(0.183,1.391)	(0.179,1.392)	(0.175,1.393)	(0.166,1.395)	(0.162,1.396)	(0.157,1.397)	(0.152,1.398)	(0.147,1.399)	(0.142,1.4)	(0.137,1.401)	(0.132,1.402)	(0.126,1.403)	(0.12,1.404)	(0.114,1.405)	(0.108,1.406)	(0.101,1.407)	(0.094,1.408)	(0.086,1.409)	(0.077,1.41)	(0.067,1.411)	(0.056,1.412)	(0.041,1.413)	(0.04,1.413)	(0.039,1.413)	(0.038,1.413)	(0.037,1.413)	(0.017,1.414)	(0.016,1.414)	(0.015,1.414)	(0.014,1.414)	(0.013,1.414)	(0.012,1.414)	(0.011,1.414)	(0.01,1.414)	(0.009,1.414)	(0.008,1.414)	(0.007,1.414)	(0.006,1.414)	(0.005,1.414)	(0.004,1.414)	(0.003,1.414)	(0.002,1.414)	(0.001,1.414)	(0,1.414)	(0,1)	(1,1)};	
    \addplot[ fill=violet, 
                    fill opacity=0.2, color=violet]
    coordinates {(1,1)	(1.414,1.414)	(1.414,3)	(1,3)	(1,1)};		
    \addplot[ fill=violet, 
                    fill opacity=0.2, color=violet]
    coordinates {(1,1)	(3,1)	(3,1.414)	(1.414,1.414)	(1,1)};
    
     \node[black] at (axis cs:0.25,0.25){${\cT}1_1$}; 
  \node[black] at (axis cs:2,0.25){${\cT}1_2$}; 
   \node[black] at (axis cs:0.25,2){${\cT}1_3$}; 
    \node[black] at (axis cs:2.25,2.25){${\cT}1_4$}; 
    \node[black] at (axis cs:1.2,0.25){${\cT}2_1$};
    \node[black] at (axis cs:0.25,1.2){${\cT}2_2$};
    \node[black] at (axis cs:1.2,2){${\cT}2_3$};
    \node[black] at (axis cs:2,1.2){${\cT}2_4$};
     \end{axis}
     \end{tikzpicture}
}
\caption{\label{fig1} Channel conditions where sum capacity is obtained for the $2$-user GIC, $P_1=P_2=1$.}
\label{fig:2-user}
\end{figure}
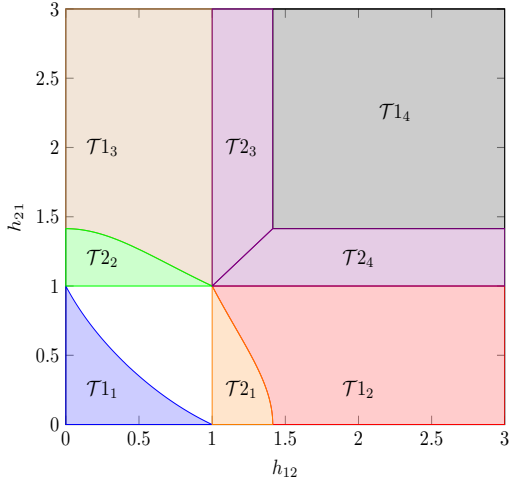

In Fig. \ref{fig:2-user}, $\mathcal{T}1_i$ denotes the region given by scheme $i$ (corresponds to $i^{th}$ row in the Table) in Table \ref{table:shk2user1} and $\mathcal{T}2_i$ denotes the region given by scheme $i$ in Table \ref{table:shk2user2}.

\begin{table*}
\centering
\begin{tabular}{|p{1.4cm} | p{2.3cm} | p{2.7cm}| p{1cm}|} 
  \hline
  S-HK scheme & Optimality conditions & Sum capacity & Matches \\ \hline
 $I(1)=\{\}$ & $h_{21} \leq 1, h_{12} \geq 1$, & $\frac{1}{2} \log(1+P_1+h_{12}^2 P_2)$ & Thm. 10 in \cite{motkha09} \\
 $I(2)=\{1\}$ & $h_{12}^2 \leq \frac{1+P_1}{1+h_{21}^2 P_1}$  &  & \\
 \hline
 $I(1)=\{2\}$ & $h_{12} \leq 1, h_{21} \geq 1 $, & $\frac{1}{2} \log(1+P_2+h_{21}^2 P_1)$ & Thm. 10 in \cite{motkha09}\\
 $I(2)=\{\}$ &  $h_{21}^2 \leq \frac{1+P_2}{1+h_{12}^2 P_2}$ &  &\\ 
 \hline
 $I(1)=\{\}$ & $1 \leq h_{12}^2 \leq 1+P_1 $, & $\frac{1}{2} \log(1+P_1+h_{12}^2P_2)$ & \\
 $I(2)=\{\}$ & $P_1+h_{12}^2P_2 \leq $ & &\\
 & ~~~~~~~$P_2+h_{21}^2P_1$ & & \cite{Sat81}\\
  \cline{2-3}
& $1 \leq h_{21}^2 \leq 1+P_2 $, & $\frac{1}{2} \log(1+P_2+h_{21}^2P_1)$ & \\
 & $P_2+h_{21}^2P_1 \leq $ & & \\
  & ~~~~~~~$P_1+h_{12}^2P_2$ & &\\
 \hline
\end{tabular}
\caption{Optimality conditions and sum capacity of S-HK schemes for a 2 user GIC using Theorems \ref{thm:achGIC}, \ref{thm:convjoint}.}
\label{table:shk2user2}
\end{table*}

\section{Partially connected GIC results}
\label{app:partialres}
We specialize our sum capacity results to cyclic, cascade, and many-to-one GICs, which are special cases of GIC. For cyclic channels, we specialize the two channel conditions for general S-HK schemes to get two new sum capacity results. For 3 user cascade channels, sum capacity results were derived in \cite{liuerk11}, but we derive sum capacity results for a general $K$ user cascade channels. For the many-to-one and one-to-many GICs, we can recover the results derived in \cite{GnaChaBha17}.

\subsection{Cyclic GIC}
\label{app:cyclic}
We use the following channel model for cyclic GIC
\[y_k= x_k+ h_{k+1} x_{k+1}+z_k, \forall k \in[K]\]
Where modulo $K$ is assumed over the indices.

\begin{result}
For a cyclic GIC, satisfying the following conditions for some sets $ I_1, D_1 \subseteq [K]$ and $I_1 \cup D_1= [K]$
\begin{align} 
\frac{h_{i+1}^2(1+Q_{i+1})^2}{\rho_{i+1}^2}  & \leq 1-\rho_i^2, \ \forall i \in I_1 \label{eqn:convcyc}\\
h_{i+1}^2 (1+Q_{i+1}) & \geq 1+P_i , \forall i \in D_1 \label{eqn:achcyc}
\end{align}
where 
\[Q_{i} = \left\{
  \begin{array}{lr}
    h_{i+1}^2P_{i+1} & : i \in I_1\\
    0 & : \mbox{else}
  \end{array}
\right.
\]
the sum capacity is given by
\begin{equation}
C_{sum} = \underset{i=1}{\overset{K}{\sum}} \frac{1}{2} \log \left[1+\frac{P_i}{1+Q_i} \right] \label{eqn:sumcapacitycyc}
\end{equation}
and the sum capacity is achieved by S-HK scheme defined by $I(i)=\phi , \ \forall i \in D_1$ and $I(i)=\{i+1 \},\  \forall i \in I_1$.
\end{result}
\begin{proof}
Use Theorem \ref{thm:covgen}.
\end{proof}

\begin{corollary} \label{cor:achcyclic}
For the cyclic channel, if we treat interference as noise at receivers $i \in I_1$ and decode interference at receivers $i \in D_1= [K] \backslash I_1$, then the achievable sum rates  given by
\begin{align}
S & \leq \frac{1}{2}\sum_{j \in \mathcal{J}_1} \log (1+P_j+h_{j+1}^2 P_{j+1}) \nonumber\\
&~ ~ ~ +\frac{1}{2} \sum_{j \in \mathcal{J}_2} \log (1+c_j P_j) \nonumber\\
& \forall  \mathcal{J}_1 \subseteq D_1, \mbox{such that if}  \ i \in  \mathcal{J}_1, \ \mbox{then}, \ i+1 \notin \mathcal{J}_1 \nonumber\\
&  \mathcal{J}_2= [K] \backslash \{i,i+1: i \in \mathcal{J}_1   \} \label{eqn:sumrateachcyclic1}\\
2S & \leq \frac{1}{2} \sum_{j=1}^K \log  (1+P_j+h_{j+1}^2 P_{j+1}), \mbox{if} \ D_1=[K] \label{eqn:sumrateachcyclic2}
\end{align}
where,  for all  $i=1,2,\cdots,K$
\[c_i = \left\{
  \begin{array}{lr}
  \min \left\{h_i^2, \frac{1}{1+h_{i+1}^2 P_{i+1}} \right\}     & : i-1 \in D_1,\ i \in I_1\\
  \min \{h_i^2, 1 \}     & : i-1 \in D_1,\ i \in D_1\\
     \frac{1}{1+h_{i+1}^2 P_{i+1}}  & : i-1 \in I_1,\ i \in I_1\\
      1   & : i-1 \in I_1,\ i \in D_1
  \end{array}
\right.
\]
are achievable.
\end{corollary}
\begin{proof}
Use Theorem \ref{thm:achGIC}.
\end{proof}

\begin{result}
For the cyclic GIC, let $I_1, D_1 \subseteq [K]$ such that $I_1 \cup D_1= [K]$ and let some  $\{k \} \in D_1$ , if the channel satisfies the following conditions 
\begin{IEEEeqnarray}{lcr}
 \frac{h_{i+1}^2(1+Q_{i+1})^2}{\rho_{i+1}^2}   \leq 1-\rho_i^2, \ \forall i \in I_1 ,\label{eqn:convcyclic1}\\
h_{k+1} \geq 1 \label{eqn:convcyclic2}\\
(1+P_{k}+h_{k+1}^2 P_{k+1}) \prod_{\underset{j \neq \{k.k+1\}}{j=1}}^{K} \left(1+\frac{P_j}{1+Q_j}\right) \leq  \nonumber\\
~ ~ \prod_{j \in \mathcal{J}_1} (1+P_j+h_{j+1}^2 P_{j+1}) \prod_{j \in \mathcal{J}_2}(1+c_j P_j), \nonumber\\
  \forall  \mathcal{J}_1 \subseteq D_1, \mbox{such that if}  \ i \in  \mathcal{J}_1,\ \mbox{then}, \ i+1 \notin \mathcal{J}_1, \nonumber\\
 \mathcal{J}_2= [K] \backslash \{i,i+1: i \in \mathcal{J}_1   \} \label{eqn:achcyclic}\\
 (1+P_{k}+h_{k+1}^2 P_{k+1}) \prod_{\underset{j \neq \{k.k+1\}}{j=1}}^{K} \left(1+\frac{P_j}{1+Q_j}\right) \leq  \nonumber\\
~ ~ \prod_{j \in [K]} \frac{1}{2} (1+P_j+h_{j+1}^2 P_{j+1})   \ \mbox{if} \ D_1=[K]\label{eqn:achcyclic1}
\end{IEEEeqnarray}
where 
\[Q_{i} = \left\{
  \begin{array}{lr}
    h_{i+1}^2P_{i+1} & : i \in I_1\\
    0 & : \mbox{else}
  \end{array}
\right.
\]
the sum capacity is given by
\begin{IEEEeqnarray}{lcr}
C_{sum}= \frac{1}{2} \log (1+P_{k}+h_{k+1}^2 P_{k+1}) \nonumber\\
~~~~~~~~+ \sum_{\underset{j \neq \{k,k+1\}}{j=1}}^{K} \frac{1}{2} \log (1+\frac{P_j}{1+Q_j}) \nonumber
\end{IEEEeqnarray}
where $c_i, \forall i \in [K]$ is defined as in corollary \ref{cor:achcyclic}.
\end{result}
\begin{proof}
Use Theorem \ref{thm:convjoint}, to get the converse conditions (\ref{eqn:convcyclic1}), (\ref{eqn:convcyclic2}) and use corollary \ref{cor:achcyclic}, to get the achievability conditions (\ref{eqn:achcyclic}).
\end{proof}

\subsection{Cascade GIC}
\label{app:cascade}
We use the following channel model for cascade GIC
\begin{align*}
y_k &= x_k+ h_{k+1} x_{k+1}+z_k, \forall k \in \{1,2,\cdots,K-1\} \\
y_K &=x_K+z_K
\end{align*}

\begin{result}
For the cascade GIC, satisfying the following conditions for some sets $ I_1, D_1 \subseteq [K]$ and $I_1 \cup D_1 \cup \{K\}= [K]$ and $\{K\} \notin  I_1, D_1 $
\begin{align} 
\frac{h_{i+1}^2(1+Q_{i+1})^2}{\rho_{i+1}^2}  & \leq 1-\rho_i^2, \ \forall i \in I_1 \label{eqn:convcyc}\\
h_{i+1}^2 (1+Q_{i+1}) & \geq 1+P_i , \forall i \in D_1 \label{eqn:achcyc}
\end{align}
where 
\[Q_{i} = \left\{
  \begin{array}{lr}
    h_{i+1}^2P_{i+1} & : i \in I_1\\
    0 & : \mbox{else}
  \end{array}
\right.
\]
the sum capacity is given by
\begin{equation}
C_{sum} = \underset{i=1}{\overset{K}{\sum}} \frac{1}{2} \log \left[1+\frac{P_i}{1+Q_i} \right] \label{eqn:sumcapacitycyc}
\end{equation}
\end{result}
\begin{proof}
We get result by taking $I(i)=\phi , \ \forall i \in I_1$ and $I(i)=\{i+1 \},\  \forall i \in D_1$ in Theorem \ref{thm:covgen}.
\end{proof}

\begin{corollary} \label{cor:achcascade}
For the cascade channel, if we treat interference as noise at receivers $i \in I_1$ and decode interference at receivers $i \in D_1= \{1,2,\cdots,K-1\} \backslash I_1$, then the sum rates given by
\begin{align}
S & \leq \frac{1}{2} \sum_{j \in \mathcal{J}_1} \log (1+P_j+h_{j+1}^2 P_{j+1}) \\
& +\frac{1}{2} \sum_{j \in \mathcal{J}_2} \log (1+e_j P_j) \nonumber\\
& \forall  \mathcal{J}_1 \subseteq D_1, \mbox{such that if}  \ i \in  \mathcal{J}_1, \ \mbox{then}, \ i+1 \notin \mathcal{J}_1 \nonumber\\
&  \mathcal{J}_2= [K] \backslash \{i,i+1: i \in \mathcal{J}_1   \} \label{eqn:sumrateachcascade}
\end{align}
where,  for all  $i=1,2,\cdots,K-1$
\[e_i = \left\{
  \begin{array}{lr}
  \min \left\{h_i^2, \frac{1}{1+h_{i+1}^2 P_{i+1}} \right\}     & : i-1 \in D_1,\ i \in I_1\\
  \min \{h_i^2, 1 \}     & : i-1 \in D_1,\ i \in D_1\\
     \frac{1}{1+h_{i+1}^2 P_{i+1}}  & : i-1 \notin D_1,\ i \in I_1\\
      1   & : i-1 \notin D_1,\ i \in D_1
  \end{array}
\right.
\]
and $e_K=1$\\
are achievable.
\end{corollary}

\begin{result}
For the cascade GIC, treating interference as noise at receivers $i \in I_1$ and decoding interference at receivers $i \in D_1$ (assuming $\{K-1\} \in D_1$) is optimal if the channel satisfies the following conditions 
\begin{IEEEeqnarray}{lcr}
 \frac{h_{i+1}^2(1+Q_{i+1})^2}{\rho_{i+1}^2}   \leq 1-\rho_i^2, \ \forall i \in I_1 ,\label{eqn:convcascade1}\\
    h_K  \geq 1 \label{eqn:convcascade2}\\
(1+P_{K-1}+h_{K}^2 P_{k}) \prod_{j=1}^{K-2} \left(1+\frac{P_j}{1+Q_j}\right) \leq  \nonumber\\
~ ~ \prod_{j \in \mathcal{J}_1} (1+P_j+h_{j+1}^2 P_{j+1}) \prod_{j \in \mathcal{J}_2}(1+e_j P_j), \nonumber\\
  \forall  \mathcal{J}_1 \subseteq D_1, \mbox{such that if}  \ i \in  \mathcal{J}_1,\ \mbox{then}, \ i+1 \notin \mathcal{J}_1, \nonumber\\
 \mathcal{J}_2= [K] \backslash \{i,i+1: i \in \mathcal{J}_1   \} \label{eqn:achcascade}
\end{IEEEeqnarray}
where 
\[Q_{i} = \left\{
  \begin{array}{lr}
    h_{i+1}^2P_{i+1} & : i \in I_1\\
    0 & : \mbox{else}
  \end{array}
\right.
\]
and the sum capacity is given by
\begin{IEEEeqnarray}{lcr}
C_{sum}= \frac{1}{2} \log (1+P_{K-1}+h_{K}^2 P_{K}) \nonumber\\
~~~~~~~~~~+ \sum_{j=1}^{K-2} \frac{1}{2} \log (1+\frac{P_j}{1+Q_j}) 
\end{IEEEeqnarray}
where $e_i, \forall i \in [K]$ is defined as in corollary \ref{cor:achcascade}.
\end{result}
\begin{proof}
Use Theorem \ref{thm:convjoint}, to get the converse conditions (\ref{eqn:convcascade1}), (\ref{eqn:convcascade2}) and use corollary \ref{cor:achcascade}, to get the achievability conditions (\ref{eqn:achcascade}).
\end{proof}

\subsection{Many-to-one GIC}
\label{app:many21}
 Channel model for many-to-one IC is given by
 \begin{align}
y_1 &=x_1+ \sum_{j=2}^K h_i x_i+z_i \nonumber\\
y_i& =x_i+z_i, \ \forall i=2,3,\cdots,K \label{mod:manytoone}
\end{align} 

\begin{result}
For a many-to-one channel, satisfying the following conditions 
\begin{IEEEeqnarray}{lcr}
\sum_{j=k+1}^{K} h_j^2 \leq 1 \label{eqn:convmanyone}\\
\underset{i \in \cB-\cN}\prod (1+P_{i}).(1+\underset{j=k+1}{\overset{K}\sum} {h_{j}^2}P_{j} + P_1) \leq \nonumber\\
~1+\underset{i\in \cB-\cN}\sum h_{i}^2 P_{i} + P_{1}
, \forall \cN \subset \cB, \cN \ne \cB,\label{eqn:achmanyone}
\end{IEEEeqnarray}
where $\cB$ = $\{2,3,\hdots,k\}$ , $k \in \{1,2,..,K\}$, the sum capacity is given by 
then the sum capacity is given by 
\begin{IEEEeqnarray}{lcr}
C_{sum} = \frac{1}{2} \log \left( 1+ \frac{P_1}{1+ \sum_{j=k+1}^{K} h_j^2 P_j}\right)\nonumber\\
 ~~~~~~~ + \sum_{i=2}^K \frac{1}{2} \log(1+P_i) \label{eqn:sumratemanyone}
\end{IEEEeqnarray}
\end{result} 
\begin{proof}
 From Theorem \ref{thm:covgen}, taking $I(1)=\{k+1,\cdots,K \}$, we get the required sum capacity if the channel satisfies  the conditions (\ref{eqn:convmanyone}) and also
\begin{equation*}
\sum_{j=k+1}^{K} \frac{h_j^2}{\rho_j^2} \leq 1-\rho_1^2
\end{equation*}
for some $\rho_i \in [0,1], \ i=1,2,\cdots,K$. \\
Choose $\rho_1=0$ and $\rho_j=1 \ \forall j =k+1,\cdots,K$ to get the condition (\ref{eqn:convmanyone}). 
\end{proof}

\begin{result}
For the K-user Gaussian many-to-one IC satisfying the following channel conditions:
\begin{IEEEeqnarray}{lcr}
\underset{i\in \cN}\prod (1+P_{i})\left( 1+P_{1}+\underset{i=k+1}{\overset{K}\sum} h_{i}^2 P_{i}+\underset{i\in \cB-\cN}\sum h_{i}^2 P_{i}\right) \nonumber\\  
\geq \label{eqn:achmik1}\prod_{i=2}^{k-1} (1+P_{i})(1+P_{1}+\sum_{j=k}^{K}h_{j}^2 P_{j})\\
\forall {\cN \subseteq \cB}, {\cN} \neq \{2,3,..,k-1\}\ and\ {\cB} = \{2,3,...k\} \nonumber  \\	
\underset{i=k+1}{\overset{K}\sum} h_{i}^2 \leq 1- \rho^2, \mbox {~~} \rho h_k = 1 + \underset{i=k+1}{\overset{K}\sum}h_{i}^2 P_{i} \label{eqn:convmik1}
\end{IEEEeqnarray}
the sum capacity is given by
\begin{IEEEeqnarray}{lcr}\label{eqn:sumcapmik1}
C_{sum} = \underset{\underset{i\neq k}{i=2}}{\overset{K}\sum} 
\frac{1}{2} \log(1+P_{i}) \nonumber\\
~~~~~~~+ \frac{1}{2} \log\Vast(1+ \frac{P_{1}+h_{k}^2P_{k}}
{1+\underset{i=k+1}{\overset{K}\sum}h_{i}^2P_{i}}\Vast)
\end{IEEEeqnarray}
\end{result}
\begin{proof}
(Converse) From Theorem \ref{thm:convjoint}, taking $G(i)=\{k+1,\cdots,K\}$ , we get the required outer bound when (\ref{eqn:convmik1}) is satisfied. \\
(Achievability) Using theorem \ref{thm:achGIC} with $I(i)=\{k+1,\cdots,K\}$,
 we get the following achievable sum rates 
\begin{multline}\label{eqn:manyoneachsumrates}
S \leq \frac{1}{2}\underset{i=k+1}{\overset{K}\sum} \log(1 + P_{i}) +\frac{1}{2} \underset{i \in \cM}\sum \log(1+m_i P_{i}) \\+ \frac{1}{2} \ \log\left(1 + \frac{P_{1} + \underset{i \in \cB-\cM}\sum h_{i}^2P_{i}}{1+ \underset{i=k+1}{\overset{K}\sum} h_{i}^2 P_{i} }\right) , \forall \cM\subseteq \cB. 
\end{multline}
where $\mathcal{B} =\{2,3,\cdots,k\}$ and \\
$m_i=\min \left\{1, \frac{h_i^2}{1+ \underset{j=k+1}{\overset{K}{\sum}}h_j^2 P_j}  \right \}$.\\
Among these sum rates, we want the sum rate with  $\mathcal{M}=\mathcal{B} \backslash \{k\}$, and $m_i=1$, $\forall i \in \mathcal{B} \backslash \{k\} $ to be dominant. We get the conditions (\ref{eqn:achmik1}), for the sum rate with $\mathcal{M}=\mathcal{B} \backslash \{k\}$ to be dominant assuming $m_i=1, \forall i \in \{2,3,\cdots,k\}$. Given the converse conditions (\ref{eqn:convmik1})and (\ref{eqn:achmik1}), conditions with $m_1=\frac{h_i^2}{1+ \underset{j=k+1}{\overset{K}{\sum}}h_j^2 P_j}$ are always redundant.
\end{proof}

\section{Examples}
\label{app:examples}
In this section, we will give some examples of finding the two set of channel conditions under sum capacity is achieved by a S-HK schemes using Theorem \ref{thm:covgen}, \ref{thm:achGIC}, \ref{thm:convjoint}.

\begin{example}
In this example, using Theorem \ref{thm:covgen} we will find the first set of channel conditions under which sum capacity is achieved by a S-HK scheme. Consider a 3-user GIC with S-HK scheme given by  $I(1)=\{2\}$, $I(2)=\{3\}$, $I(3)=\{\}$. Inequalities (\ref{eqn:fin1}), (\ref{eqn:fin2}) gives the same set of conditions given by
\begin{align}
h_{12}^2(1+h_{23}^2P_3) & \leq \rho_2^2(1-\rho_1^2) \label{eqn:exm1cond1}\\
h_{23}^2 & \leq (1-\rho_2^2)
\end{align}
for some $\rho_2, \rho_1 \in (0,1)$. In inequality (\ref{eqn:fin3}), for $i=1$, $\mathcal{J}=\{3\}$; for $i=2$, $\mathcal{J}=\{1\}$; and for $i=3$, $\mathcal{J}$ can be $\{1\}$, $\{2\}$, and $\{1,2\}$. Therefore,  (\ref{eqn:fin3}) gives the set of conditions
\begin{align}
1+P_1+h_{12}^2P_2 & \leq h_{13}^2\\
1+P_2+h_{23}^2P_3 & \leq h_{21}^2(1+h_{12}^2P_2)\\
(1+P_3) & \leq h_{31}^2(1+h_{12}^2P_2)\\
(1+P_3) & \leq h_{32}^2(1+h_{23}^2 P_3)\\
(1+\frac{P_1}{1+Q_1})(1+\frac{P_2}{1+Q_2})  & \leq \left(1+\frac{h_{31}^2P_1+h_{32}^2P_2}{1+P_3} \right) \label{eqn:exm1cond2}
\end{align} 
where $Q_1=h_{12}^2P_2$, $Q_2=h_{23}^2P_3$, $Q_3=0$. Under conditions (\ref{eqn:exm1cond1})-(\ref{eqn:exm1cond2}), sum capacity is achieved by S-HK scheme with $I(1)=\{2\}$, $I(2)=\{3\}$, $I(3)=\{\}$ and the sum capacity is given by
\begin{equation*}
C_{sum} = \sum_{i=1}^{3} \frac{1}{2} \log \left[1+\frac{P_i}{1+Q_i} \right].
\end{equation*}
\end{example}

\begin{example}
In this example, we will find achievable sum rates in (\ref{eqn:achsum-rategenIC}) for a S-HK scheme. Consider a 3-user GIC with S-HK scheme given by $I(1)=\{\}$, $I(2)=\{1,3\}$, $I(3)=\{2\}$ which implies $D(1)=\{2,3\}$, $D(2)=\{\}$, $D(3)=\{1\}$. From Theorem \ref{thm:achGIC}, $l$ can be 1,2,3 and $\mathcal{J}_1 \subseteq \{1,2,3\}$, $i.e.$, $\mathcal{J}_1 $ can be \{\}, \{1\}, \{2\}, \{3\}, \{1,2\},\{1,3\},\{2,3\}, \{1,2,3\}. $\mathcal{J}_2 \subseteq \{2\}$, $\mathcal{J}_3 \subseteq \{1,3\}$. For $l=1,2,3$, the possible sets of $\mathcal{J}_1$, $\mathcal{J}_2$, $\mathcal{J}_3$ such that $\underset{i \in {[K]}}{\bigcup} \mathcal{J}_i= {\cal S}_l$ are given in table \ref{table:ex2}.

\begin{table}
\centering
\begin{tabular}{|p{1cm} | p{1cm} | p{1cm}| p{1cm}|} 
  \hline
  $l$ & $\mathcal{J}_1$ & $\mathcal{J}_2$ & $\mathcal{J}_3$ \\ \hline
 $l=1$ & $\{\}$ & $\{2\}$ & $\{1,3 \}$ \\
 \cline{2-4}
 & $\{1\}$ & $\{2\}$ & $\{3 \}$ \\
 \cline{2-4}
 & $\{2\}$ & $\{\}$ & $\{1,3 \}$ \\
 \cline{2-4}
 & $\{3\}$ & $\{2\}$ & $\{1 \}$ \\
 \cline{2-4}
 & $\{1,2\}$ & $\{\}$ & $\{3 \}$ \\
 \cline{2-4}
 & $\{1,3\}$ & $\{2\}$ & $\{ \}$ \\
 \cline{2-4}
 & $\{2,3\}$ & $\{\}$ & $\{1 \}$ \\
 \cline{2-4}
 & $\{1,2,3\}$ & $\{\}$ & $\{ \}$ \\
 \hline
  $l=2$ & $\{1,2,3\}$ & $\{2\}$ & $\{1,3 \}$ \\
  \hline
\end{tabular}
\caption{Set of $\mathcal{J}_i$ such that $\underset{i \in {[K]}}{\bigcup} \mathcal{J}_i= {\cal S}_l$ for S-HK with $I(1)=\{\}$, $I(2)=\{1,3\}$, $I(3)=\{2\}$.}
\label{table:ex2}
\end{table}
Achievable sum rates are given by
\begin{align*}
S & \leq \frac{1}{2} \log \left[1+\frac{P_2}{1+Q_2} \right] +\frac{1}{2} \log \left[1+\frac{h_{31}^2P_1+P_3}{1+Q_3} \right]\\
S & \leq \frac{1}{2} \log (1+P_1)+\frac{1}{2} \log \left[1+\frac{P_2}{1+Q_2} \right] \\
&~~~ +\frac{1}{2} \log \left[1+\frac{P_3}{1+Q_3} \right] \\
S & \leq \frac{1}{2} \log (1+h_{12}^2P_2) +\frac{1}{2} \log \left[1+\frac{h_{31}^2P_1+P_3}{1+Q_3} \right]\\
S & \leq \frac{1}{2} \log (1+h_{13}^2P_3) +\frac{1}{2} \log \left[1+\frac{P_2}{1+Q_2} \right] \\
&~~~ +\frac{1}{2} \log \left[1+\frac{h_{31}^2P_1}{1+Q_3} \right]\\
S & \leq \frac{1}{2} \log (1+P_1+h_{12}^2P_2)+\frac{1}{2} \log \left[1+\frac{P_3}{1+Q_3} \right]\\
S & \leq \frac{1}{2} \log (1+P_1+h_{13}^2P_3)+\frac{1}{2} \log \left[1+\frac{P_2}{1+Q_2} \right]\\
S & \leq \frac{1}{2} \log (1+h_{12}^2P_2+h_{13}^2P_3)+\frac{1}{2} \log \left[1+\frac{h_{31}^2P_1}{1+Q_3} \right]\\
S & \leq \frac{1}{2} \log (1+P_1+h_{12}^2P_2+h_{13}^2P_3)\\
2S & \leq \frac{1}{2} \log (1+P_1+h_{12}^2P_2+h_{13}^2P_3)\\
&+\frac{1}{2} \log \left[1+\frac{P_2}{1+Q_2} \right]+ +\frac{1}{2} \log \left[1+\frac{h_{31}^2P_1+P_3}{1+Q_3} \right].
\end{align*}
where $ Q_2=h_{21}^2P_1+h_{23}^2P_3$, $Q_3=h_{23}^2P_2$. Depending on the channel and power constraints one of the above inequalities will be dominant.

\end{example}
  
\begin{example}
In this example, using Theorem \ref{thm:achGIC}, \ref{thm:convjoint} we will find the second set of channel conditions under which sum capacity is achieved by a S-HK scheme. Consider a 3-user GIC with S-HK scheme given by  $I(1)=\{3\}$, $I(2)=\{1,3\}$, $I(3)=\{\}$. Let $m=1$, $k=2$. Here $m$, $k \notin I(i)$, $i \in \{1,3\}$. First we will find the converse conditions or the conditions under which sum rate
\begin{IEEEeqnarray}{lcr}
S \leq & \frac{1}{2} \log \left[1+\frac{(P_1+h_{12}^2P_2)}{1+Q_1} \right]  +  \frac{1}{2} \log (1+P_3)  \label{eqn:upperb2}
\end{IEEEeqnarray}
is an upper bound. For $I(1)=\{3\}$, $I(2)=\{1,3\}$, $I(3)=\{\}$, inequalities (\ref{eqn:jointfin1}),(\ref{eqn:jointfin4}) gives the same set of conditions
\begin{align}
h_{13}^2 & \leq \rho_3^2 (1-\rho_2^2) \label{eqn:concondt1}.
\end{align}
(\ref{eqn:jointfin2}) does not give any condition. (\ref{eqn:jointfin3}) implies
\begin{align}
\rho_2 h_{12} &=1+h_{13}^2 P_3 \label{eqn:concondt2}
\end{align}
Combining (\ref{eqn:concondt1}), (\ref{eqn:concondt2}), sum rate in (\ref{eqn:upperb2}) is an upper bound for all the channels satisfying
\begin{equation}
h_{13}^2+ \left(\frac{1+h_{13}^2P_3}{h_{12}} \right)^2 \leq 1 \label{eqn:concondtfin}
\end{equation}
For achievablity conditions, first we will find all achievable sum rates of the S-HK scheme using Theorem \ref{thm:achGIC}. Observe that $\mathcal{J}_1 \subseteq \{1,2\}$, $\mathcal{J}_2 \subseteq \{2\}$, $\mathcal{J}_3 \subseteq \{1,2,3\}$. For $l=1,2,3$, the possible sets of $\mathcal{J}_1$, $\mathcal{J}_2$, $\mathcal{J}_3$ such that $\underset{i \in {[K]}}{\bigcup} \mathcal{J}_i= {\cal S}_l$ are given in table \ref{table:ex3}.
 
\begin{table}
\centering
\begin{tabular}{|p{1cm} | p{1cm} | p{1cm}| p{1cm}|} 
  \hline
  $l$ & $\mathcal{J}_1$ & $\mathcal{J}_2$ & $\mathcal{J}_3$ \\ \hline
 $l=1$ & $\{\}$ & $\{\}$ & $\{1,2,3 \}$ \\
 \cline{2-4}
  & $\{\}$ & $\{2\}$ & $\{1,3 \}$ \\
 \cline{2-4}
 & $\{1\}$ & $\{\}$ & $\{2,3 \}$ \\
 \cline{2-4}
 & $\{1\}$ & $\{2\}$ & $\{3 \}$ \\
 \cline{2-4}
 & $\{2\}$ & $\{\}$ & $\{1,3 \}$ \\
 \cline{2-4}
 & $\{1,2\}$ & $\{\}$ & $\{3 \}$ \\
 \hline
\end{tabular}
\caption{Set of $\mathcal{J}_i$ such that $\underset{i \in {[K]}}{\bigcup} \mathcal{J}_i= {\cal S}_l$ for S-HK with $I(1)=\{3\}$, $I(2)=\{1,3\}$, $I(3)=\{\}$.}
\label{table:ex3}
\end{table}
Achievable sum rates for S-HK with $I(1)=\{3\}$, $I(2)=\{1,3\}$, $I(3)=\{\}$ are given by
\begin{align}
S & \leq \frac{1}{2} \log \left[1+P_3+h_{32}^2P_2+h_{31}^2P_1 \right] \label{eqn:achcond1} \\
S & \leq \frac{1}{2} \log \left[1+ \frac{P_2}{1+Q_2} \right]+ \frac{1}{2} \log \left[1+P_3+h_{31}^2 P_1 \right]\\
S & \leq \frac{1}{2} \log \left[1+\frac{P_1}{1+Q_1} \right]+ \frac{1}{2} \log \left[1+P_3+h_{32}^2P_2 \right]\\
S & \leq \frac{1}{2} \log \left[1+\frac{P_1}{1+Q_1}\right]+ \frac{1}{2} \log \left[1+\frac{P_2}{1+Q_2} \right]\\
& ~~~~~+\frac{1}{2}  \log(1+P_3)\\
S & \leq \frac{1}{2} \log \left[1+\frac{h_{12}^2P_2}{1+Q_1} \right]+ \frac{1}{2} \log \left[1+P_3+h_{31}^2P_1 \right]\\
S & \leq \frac{1}{2} \log \left[1+\frac{P_1+h_{12}^2P_2}{1+Q_1} \right]+ \frac{1}{2} \log (1+P_3 ) \label{eqn:achcond6}
\end{align}
where $Q_1= h_{13}^2P_3$, $Q_2=h_{21}^2P_1+h_{23}^2 P_3$. We want (\ref{eqn:achcond6}) to be dominant among the inequalities (\ref{eqn:achcond1})-(\ref{eqn:achcond6}) which gives the conditions
\begin{IEEEeqnarray}{lCr}
(1+P_3)(P_1+h_{12}^2P_2) \leq (h_{32}^2P_2+h_{31}^2P_1)(1+Q_1) \label{eqn:achcondfin1}\\
(1+P_3) \left(\frac{P_1+h_{12}^2P_2}{1+Q_1} \right) \leq\frac{(1+P_3+h_{31}^2P_1)P_2}{1+Q_2} \nonumber \\
~~~~~~~~~~~~~~~~~~~~~~~~~~~~~~~~~~~~ + h_{31}^2P_1\\
h_{12}^2 (1+P_3) \leq h_{32}^2(1+P_1+Q_1)\\
h_{12}^2(1+Q_2) \leq (1+P_1+Q_1) \\
(1+P_3) \leq (1+Q_1+h_{12}^2P_2)h_{31}^2 \label{eqn:achcondfin5}
\end{IEEEeqnarray}
\end{example}
Therefore, under conditions (\ref{eqn:achcondfin1})-(\ref{eqn:achcondfin5}) and (\ref{eqn:concondtfin}), sum capacity is achievable by S-HK scheme with $I(1)=\{3\}$, $I(2)=\{1,3\}$, $I(3)=\{\}$ and the sum capacity is given by (\ref{eqn:upperb2}).

\end{document}